\documentclass[12pt,final]{article}

 \usepackage[T1]{fontenc}
   \usepackage{fourier}
\usepackage{xcolor}
\usepackage[reqno]{amsmath}
\usepackage[makeroom]{cancel}
\usepackage{amssymb}
\usepackage{nicefrac}
\usepackage{calrsfs}
\DeclareMathAlphabet{\pazocal}{OMS}{zplm}{m}{n}

\usepackage{tikz}
\usepackage{float}
\usetikzlibrary{arrows}
\usetikzlibrary{trees,calc,snakes,fit,shapes,positioning,automata,quotes}
\usetikzlibrary{decorations.pathreplacing}
\usetikzlibrary{decorations.pathmorphing}
\usetikzlibrary{decorations.markings}
\usetikzlibrary{decorations.shapes}
\tikzset{
  treenode/.style = {align=center, inner sep=0pt, text centered,
    font=\sffamily},
  arn_n/.style = {treenode, circle, black, font=\sffamily\bfseries, draw=black,
    fill=white, text width=1.5em},
  arn_r/.style = {treenode, circle, bg, draw=bg, 
    text width=1.5em},
  arn_x/.style = {treenode, circle, purple,font=\sffamily\bfseries,
   text width=2em}
}
\tikzset{
  invisible/.style={opacity=0},
  visible on/.style={alt={#1{}{invisible}}},
  alt/.code args={<#1>#2#3}{%
    \alt<#1>{\pgfkeysalso{#2}}{\pgfkeysalso{#3}} 
  },
}
\tikzset{
    photon/.style={decorate, decoration={snake}, draw=red}}
\tikzset{electron/.style={draw=blue, thick, postaction={decorate},decoration={markings,mark=at position .75 with {\arrow[draw=blue]{>}}}}}
\tikzset{electron2/.style={draw=red,very thick, postaction={decorate},decoration={markings,mark=at position .55 with {\arrow[draw=red,very thick]{>}}}}}
\tikzset{gluon/.style={->,thick,decorate, draw= mLightBrown,
        decoration={coil,amplitude=4pt, segment length=5pt}}}
 \tikzset{gluon2/.style={thick,decorate, draw=mLightBrown,
        decoration={coil,amplitude=4pt, segment length=5pt}}}
\tikzset{nero/.style={decorate,draw=black}}
\tikzset{bianco/.style={decorate,draw=bg}}
 \tikzset{ every node/.style={inner sep=0pt,minimum size=1mm},
  nsnode/.style={draw,circle,black,minimum size=2mm},
  asnode/.style={draw,circle,myblue,fill=myblue},
  bsnode/.style={draw,circle,black,fill=black,nero},
  csnode/.style={draw,circle,red,fill=red, minimum size=2mm},
  every fit/.style={inner sep=-1.5pt,text width=1cm}  }
\usepackage{soul}
\usepackage{bbm}
\usepackage{amsthm}
\makeatletter
\def\thm@space@setup{%
  \thm@preskip=\parskip \thm@postskip=0pt
}
\makeatother
\makeatletter
\newcommand{\leqnomode}{\tagsleft@true}
\newcommand{\reqnomode}{\tagsleft@false}
\makeatother
\usepackage{setspace}
\usepackage{amsfonts}
\usepackage{graphicx}
\usepackage{fnpct}
\usepackage{apptools}
\usepackage[justification=centering]{caption}
\usepackage[justification=centering,singlelinecheck=true]{subfig}
\usepackage{natbib}
\usepackage{titlesec}

\usepackage[colorlinks,linkcolor=links,citecolor=cites,urlcolor=MyDarkBlue]{hyperref}
\usepackage{cleveref}
\usepackage{xr}
\usepackage{geometry}
\usepackage{enumitem}
\usepackage[section]{placeins}
\definecolor{MyDarkBlue}{rgb}{0,0.08,0.45}
\definecolor{cites}{HTML}{324b13}
\definecolor{links}{HTML}{1a663b}
\definecolor{MyLightBlue}{cmyk}{0.1,0.8,0,0.1}
\newtheoremstyle{ex}
{1pt}
{1pt}
{}
{}
{\bfseries}
{.}
{.5em}
{}%
\newtheorem{Example}{Example}
\newtheorem{Theorem}{Theorem}
\newtheorem*{Theorem*}{Theorem}
\newtheorem{Definition}{Definition}

\newtheorem{Remark}{Remark}
\newtheorem{Corollary}{Corollary}
\newtheorem{lemma}{Lemma}
\newtheorem{Axiom}{Axiom}
\newtheorem{Assumption}{Assumption}

\newtheorem{Rule}{Rule}

\newtheorem*{axiom6restated}{\autoref{ax:axiom6}}

\title{Causality: A decision-theoretic foundation\footnote{The latest version of this paper can be found \href{https://www.dropbox.com/scl/fi/wx9b8f8yqodoiq5x1gzkq/Causality-Public.pdf?rlkey=puzpv290f5icrzivjooxxvqk8&dl=0}{here}.I wish to thank David Ahn, Arjada Bardhi, Jeff Ely, Simone Galperti, Bart Lipman, Marciano Siniscalchi, and Tristan Tomala for insightful discussions on the paper.}}
\author{Pablo Schenone\thanks{Department of Economics, Fordham University. E-mail: \href{mailto:pschenone@fordham.edu}{\texttt{pschenone@fordham.edu}}}}

\begin{document}


\begin{titlepage}

\maketitle
\begin{abstract}
We propose a decision theoretic framework that allows a decision maker to express its causal model of the world. We extend the model of \cite{savage1972foundations} by allowing the decision maker (DM) to choose policy interventions prior to choosing acts over the non-intervened variables. We define what it means for the DM's choices to express the DM’s belief that the relation between some variables is causal. We provide axioms characterizing when the DM’s causal model, as expressed through the DM’s choices, is represented as a directed acyclic graph. A final axiom characterizes when the DM’s causal model has a representation like the one in \cite{pearl1995causal}. Consequently, under this additional axiom one can apply Pearl’s results to identify the DM’s causal model from the DM’s probabilistic model. 
\end{abstract}


\end{titlepage}

\section{Introduction}

\indent An economist (Alex) studies the relationship between intellectual ability, education level, and lifetime earnings in a population. A policy maker approaches Alex for her opinion on how education and earnings are related, and whether a policy that makes college degrees compulsory would increase earnings. He is perplexed at Alex's answer:
\begin{enumerate}[label=(\roman*)]
\item If a person has a college degree they are more likely to have high earnings than if they only had a high school diploma, \label{item1}
\item however, forcing high school graduates to obtain a college degree will not increase their earnings.\label{item2}
\end{enumerate}

\indent The statements above reflect the adage ``correlation is not causation''. The first statement refers to Alex's joint beliefs over education and earnings--often modeled via probability distributions. It is a statement about the ``correlation'' (or \emph{probabilistic (in)dependence}) of the relevant variables.\footnote{Because correlation is used for both general statistical dependence and \emph{linear} statistical dependence, moving forward we use probabilistic (in)dependence to avoid confusion.} The second statement is a statement about \emph{policy relevance}: will an education policy change Alex's beliefs about potential earnings? Extrapolating policy relevance from probabilistic dependence is how economics generally understands ``causality''. 

\indent In this paper, we propose a decision-theoretic framework in which to precisely define and study causality. In the classic \cite{savage1972foundations} framework, statement \ref{item1} is defined through choices over bets defined on the education and earnings outcomes. However, the classic Savage model is not rich enough to allow for the behavior that defines statement \ref{item2}. Because comparing behavior consistent with statements \ref{item1} and \ref{item2} is how we understand causality, we first extend the Savage framework to accommodate the behavior that defines statements like \ref{item2}. In the extended framework, we provide a precise definition of causality--which we refer to as Bayesian causality--and prove a representation result. The representation result has a graph-theoretic component, which we connect to the existing literature on causal inference (see, e.g., \citealp{pearl1995causal}). Finally, we provide conditions that are necessary and sufficient for causality to be expressible purely in terms of probabilistic (in)dependence. For empirical economics, the graphical component--which \autoref{th:theorem1} axiomatizes--provides an efficient language with which to convey the causal assumptions the researcher makes. Furthermore, \autoref{th:newtheorem2} shows that causal statements may indeed be expressed purely in terms of probabilistic independence, which is useful for inferring causal relations in non-experimental settings. Below, we provide an informal overview of the model, the results, and our contributions.


\paragraph{Extending the Savage framework} Following \cite{savage1972foundations}, statement \ref{item1} above can be defined via Alex's betting behavior. Suppose Alex is offered the following bets, the payoff of which depend on the education ($E$) and lifetime earnings ($L$) of a representative citizen, Mr Kane:
\begin{figure}[H]
\centering
\subfloat[High lifetime earnings if high school]{
\begin{tabular}{c|cc}
&$L< \$ 3M$&$L\geq \$ 3M$\\
\hline
$E\neq$ High school, $a\in A$&$\$0$&$\$0$\\
$E=$ High school, $a\in A$&$\$1$&$\$0$
\end{tabular}\label{fig:hslow}}\subfloat[Low lifetime earnings if high school]{
\begin{tabular}{c|cc}
&$L< \$ 3M$&$L\geq \$ 3M$\\
\hline
$E\neq$ High school, $a\in A$&$\$0$&$\$0$\\
$E=$ High school, $a\in A$&$\$0$&$\$1$
\end{tabular}\label{fig:hshigh}}
\caption{Two conditional bets: high vs low lifetime earnings ($L$), conditional on high school education ($E$).}\label{table:1}
\end{figure}
\begin{figure}[H]
\centering
\subfloat[High lifetime earnings if college]{
\begin{tabular}{c|cc}
&$L< \$ 3M$&$L\geq \$ 3M$\\
\hline
$E\neq$ College, $a\in A$&$\$0$&$\$0$\\
$E=$ College, $a\in A$&$\$1$&$\$0$
\end{tabular}\label{fig:clow}}\hfill
\subfloat[Low lifetime earnings if college]{
\begin{tabular}{c|cc}
&$L< \$ 3M$&$L\geq \$ 3M$\\
\hline
$E\neq$ College, $a\in A$&$\$0$&$\$0$\\
$E=$ College, $a\in A$&$\$0$&$\$1$
\end{tabular}\label{fig:chigh}}
\caption{Two conditional bets: high vs low lifetime earnings ($L$), conditional on college education ($E$).}\label{table:2}
\end{figure}

\indent  \autoref{table:1} shows two bets on Mr Kane’s lifetime earnings \emph{conditional} on Mr Kane’s education. Both pay \$0 if Mr. Kane's highest degree is not high school. Instead, if Mr. Kane's highest education is high school, the act in panel \ref{fig:hslow} pays \$$1$ if Mr. Kane's earnings are low (less than $3M$) and \$$0$ otherwise, while the act in panel \ref{fig:hshigh} pays \$$1$ if earnings are high (more than $3M$) and \$$0$ otherwise. Suppose Alex chooses the low earnings bet over the high earnings bet. Alex's choice reveals that, conditional on a high school education, Alex believes earnings are more likely to be low than they are to be high. Similarly, \autoref{table:2} shows bets on Mr. Kane's lifetime earnings conditional on Mr. Kane having a college degree. Suppose Alex chooses the high earnings bet over the low earnings bet. Alex's choice reveals that, conditional on a college degree, Alex believes lifetime earnings are more likely to be high than they are to be low. Jointly, Alex's choices imply statement \ref{item1}. If we represent Alex's beliefs via probability distributions, we say that Alex behaves as if education and earnings are \emph{probabilistically dependent}.

\indent Next, we extend the classic \cite{savage1972foundations} framework so that statement \ref{item2} may also be defined in terms of Alex's choice behavior. In the Savage model, education levels are realized outside of Alex's control. In contrast, education policies are tools that influence what education levels are realized. Therefore, the Savage model has no language with which to make claims such as \ref{item2}. 

\indent We enrich Alex's choice domain so that Alex chooses both bets and \emph{policy interventions}. A policy intervention describes both which variables Alex will influence and the levels to which these variables will be set. For example, an education policy, $E=$ college, is a policy that forces education to take the value ``college'' but leaves ability and earnings unaffected. Policy interventions could be literal--the government forces Mr Kane to get a college degree, after which Alex's beliefs over the remaining variables can be elicited--or a counterfactual analysis--``if Alex were to impose a college education on Mr. Kane, how would that impact Alex's choice over bets defined on Mr. Kane's earnings?''.

\indent We now have language to express statement \ref{item2}. Suppose that for any two education policies, $e$ and $e^\prime$, and any two bets on earnings, $f$ and $g$, Alex chooses the pair $(e,f)$ over the pair $(e,g)$ if and only if Alex chooses the pair $(e^\prime,f)$ to $(e^\prime,g)$. This reveals that Alex's beliefs about earnings are insensitive to education policies. As statement \ref{item2} claims, education policies have no effect on Alex's beliefs about earnings. We refer to this as \emph{intervention independence}. 

\paragraph{Defining causality: pay for policies, or pay for conditional bets?} In the running example, before placing a bet on lifetime earnings, Alex is willing to pay a small fee to make the bets on lifetime earnings conditional on the realized education levels; however, Alex is not willing to pay for a policy that determines Mr. Kane’s education level before betting on lifetime earnings. Alex's behavior may be rationalized as follows. Alex believes the following: (a) ability is inherently informative about education, (b) ability is inherently informative about earnings, and (c) any information education contains about earnings is contained in the information ability conveys about earnings--in probability language, education and earnings are independent conditional on ability. This rationalizes why Alex is willing to pay to make bets on earnings contingent on education levels. While Alex actually values making the earnings bets contingent on ability, making the earnings bet contingent on education is still valuable because education is informative about ability, which in turn is informative about earnings. Furthermore, (a)-(c) also rationalizes Alex’s unwillingness to pay for a policy that determines Mr Kane’s education level before betting on lifetime earnings. A policy that directly chooses the education levels for everyone destroys the information education provides about ability: everyone being forced to get the same education does not mean that everyone's realized ability will be the same. Therefore, after the policy intervention, education is no longer informative about earnings because education is no longer informative about ability.

\indent The above discussion emphasizes a key difference between probabilistic and intervention independence: while probabilistic independence is symmetric, intervention independence is not. Indeed, if higher education is associated with higher ability, then higher ability must be associated to higher education. If Alex has to bet on ability, Alex would pay a small enough fee to make this bet contingent on the realized education. Symmetrically, if Alex had to bet on education, Alex should pay a small enough fee to make the bet contingent on the realized ability. However, this symmetry need not hold for how Alex's beliefs react to policy interventions: Alex might believe a policy that fixes the ability levels is inherently informative about education--after all, Alex believes that higher ability people tend to educate themselves more--whereas a policy on education is uninformative about ability levels--everyone being forced to get the same education does not mean that everyone's realized ability will be the same. 


\indent The discussion above motivates our definition of causality, which we denote \emph{Bayesian causality} (see \autoref{def:causality} in \autoref{sec:defcausal}). A variable (e.g., education) is a Bayesian cause of another variable (e.g., lifetime earnings) if a \emph{ceteris paribus} policy intervention on the presumed cause affects beliefs about the presumed consequence. Specifically, education is a \emph{Bayesian cause of} earnings if a policy that exogenously changes education levels while keeping ability constant affects Alex's beliefs over earnings. Crucially, Bayesian causality makes use of two important terms: ceteris paribus, and policy interventions. Regarding the first, we must ensure that Alex's beliefs over earnings are sensitive to education policies \emph{for each fixed ability level}. Keeping ability fixed ensures that further information on education is valuable only if education is intrinsically informative about earnings. Regarding the second, we define Bayesian causality via policy interventions rather than contingent bets because breaking the symmetry in information content is essential to differentiate \emph{cause} from \emph{consequence}.

\paragraph{Representations} Having extended the Savage model to accommodate choices over policy interventions and bets, and having defined Bayesian causality, we provide a representation of Alex's preferences over policy interventions and bets, from which we obtain a quantitative model of Bayesian causality. We proceed in two steps.

\indent First, we define what it means for a directed acyclic graph (DAG) to represent Alex's preferences over policy interventions and bets in \autoref{sec:representations}. We assume that preferences in the extended Savage framework admit a subjective expected utility representation, where Alex's subjective beliefs are represented via probability distributions.\footnote{This assumption is guaranteed by adapting standard axioms from expected utility theory.} Following the literature on graphical representations of independence, we may represent the probabilistic independence properties of the probabilities that represent Alex's subjective beliefs with either a directed or an \emph{undirected} graph (see \citealp{geiger1990identifying,lauritzen1990independence}). Because probabilistic independence is symmetric, the directed aspect of the graph is unnecessary for this aspect of the representation. Instead, as discussed in our definition of causality, intervention independence is not symmetric. The direction of the arrows is what distinguishes probabilistic independence from intervention independence. A DAG represents Alex's preferences if it represents the probabilistic and intervention independence properties of Alex's expected utility probabilities. \autoref{def:representation2} formally defines when a graph represents Alex's preference in the extended Savage framework.
%

\indent Second, \autoref{th:theorem1}  in \autoref{sec:results} connects the DAG representation with Bayesian causality. Under axioms \ref{ax:axiom3} through \ref{ax:axiom2} in \autoref{sec:axioms}, Bayesian causality imposes restrictions on how conditional independence and intervention independence are related. \autoref{th:theorem1} leverages these connections and proves the following result. In any graph that represents the expected utility probability distributions, an arrow points from a variable $i$ to a variable $j$ if and only if $i$ is a Bayesian cause of $j$. Consequently, we obtain a representation of Bayesian causality both in terms of expected utility and DAGs. That the above results holds for any representing DAG implies Bayesian causality is the only definition of causality representable by a DAG.

\paragraph{Expressing causation purely in terms of conditional independence} Lastly, \autoref{th:newtheorem2} in  \autoref{sec:statistics} shows that an \emph{additional axiom}, \autoref{ax:axiom6}, is necessary and sufficient for intervention independence to be expressible purely in terms of probabilistic dependence. Thus, under all our axioms, we need only observe Alex's behavior in the standard Savage framework to have a complete understanding of both Alex's probabilistic beliefs and her causal model. Including the choice of policy interventions is needed to define Bayesian causality, to distinguish it from probabilistic (in)dependence, and to formalize the axioms required to represent Bayesian causality. Once we write this model and assume all our axioms, Theorem \ref{th:newtheorem2} implies that we may forget these details because the traditional Savage framework summarizes all relevant information needed to express Bayesian causality. This result is of special use for empirical research, as it shows how we can calculate causal effects from conditional independence. 

\indent DAGs as a language to express causal models have been extensively used in the causal inference literature, and have recently been used in economic theory models to model misspecified beliefs. We now explore how our paper connects to these two literatures.

\paragraph{DAGs in causal inference} The paper contributes to the statistical literature on causal inference,  which often uses DAGs as a language with which to encode causal information (see, e.g., \citealp{pearl1995causal,spirtes2000causation}). 

\indent The causal inference literature generally proceeds in two steps. First, the researcher makes assumptions on how the variables of interest causally interact with each other. The exact definition of causality is left up to the researcher, with the implicit understanding it must be representable in a DAG. As Pearl states, ``\emph{The first step in this analysis is to construct a causal diagram (...) which represents the investigator's understanding of the major causal influences among measurable quantities in the domain}''. Next, some restrictions are imposed on the probability distributions of interest.  The objective of these restrictions is to connect the arrows in the DAG (which are assumed to have ``causal'' connotation, whichever way the researcher defines ``causal'') to the probabilistic independence properties implied by the DAG. The question of interest is whether the information the arrows encode--i.e., ``causal effects''--can be recovered from the (observable) probability distributions. As Pearl writes, \emph{``The purpose of the paper is not to validate or repudiate such domain-specific assumptions but, rather, to test whether a given set of assumptions is sufficient for quantifying causal effects from non-experimental data''}. 

\indent Our results contribute to this literature in three ways. First, \autoref{th:theorem1} shows that Bayesian causality is indeed one definition of causality that can be represented by a DAG, but only if our axioms hold. This provides a concrete definition of causality that is compatible with a DAG representation, something that is left out of the causal inference literature. Second, within our extended Savage framework, \autoref{th:theorem1} shows that Bayesian causality is the only definition that can be represented by a DAG. In our extended Savage framework, this implies that using DAGs to encode ``causality'' is equivalent to defining ``causality'' as ``Bayesian causality''. Lastly, \autoref{th:newtheorem2} shows that intervention independence--the key component in Bayesian causality--is expressible purely in terms of probability distributions. This is analogous to the exercise carried out in the causal inference literature. However, rather than imposing restrictions on the probability distributions themselves, we show the restrictions on choice behavior that are both necessary and sufficient to recover Bayesian causality in terms of conditional independence.

\paragraph{DAGs in economic theory} In economic theory, \cite{spiegler2016bayesian,spiegler2017data,spiegler2018can} introduces the language of DAGs to study the effects on decision making of misspecified causal models. Similar to the aforementioned causal inference literature, these papers take as given that arrows in the DAG point from cause to effect, without defining what it means for the decision maker to believe that one variable is a cause of another. Our paper complements this body of work by providing a microfoundation, in terms of a decision maker’s choice behavior, for the use of DAGs to represent a decision maker’s causal model. By pinning down the unique definition of causality that is consistent with a DAG representation–Bayesian causality–our results make transparent the definition of causality the analyst ascribes to when using DAGs to represent a decision maker’s causal model.

\indent In decision theory, \cite{karni2006subjective,karni2017states} explore models in which a decision maker can affect the states that are realized. In those papers, the primitive objects are a set of actions and a set of consequences. States of nature are defined as mappings from actions that a decision maker might take to consequences that arise from those actions; that the mapping from actions to outcomes is stochastic reflects that states are stochastically realized. A decision maker can affect the states that occur by making an appropriate choice of action. This idea is similar to our idea of a policy intervention, because an intervention policy is an action that the decision maker takes that affects the realization of states. However, the focus of Karni's paper is not on using these ideas to discuss causal effects or understand what types of models reflect normative definitions of causality. Rather, Karni focuses on obtaining subjective expected utility representations in the absence of an objectively given state space.

\indent The rest of the paper is organized as follows. \autoref{sec:dt framework} outlines the model and our extension of the Savage framework, and \autoref{sec:defcausal} formally defines Bayesian causality. Sections \ref{sec:representations} and \ref{sec:axioms} present our representations and axioms, respectively. Finally, \autoref{sec:results} contains our two representation theorems. All proofs are in the appendix.

\section{Model}\label{sec:dt framework}
In this section, we introduce our model. \autoref{sec:notation} summarizes our notational conventions, and \autoref{sec:savage-expansion} shows our extension of the Savage framework to allow for the choice of policy interventions. 

\subsection{Notation}\label{sec:notation}
\paragraph{General Notation} We use the following notation throughout the paper. The set $\pazocal{N}=\{1,...,N\}$ is a fixed set of indices. For each $\pazocal{J}\subset\pazocal{N}$, let $\{X_j:j\in\pazocal{J}\}$ be a family of sets indexed by $\pazocal{J}$. We denote by $X_{\pazocal{J}}=\Pi_{j\in\pazocal{J}}X_j$ the Cartesian product of the family and by $x_\pazocal{J}=(x_j)_{j\in\pazocal{J}}$ a typical element in $X_{\pazocal{J}}$; furthermore, we use $X$ to denote $X_\pazocal{N}$. Moreover, all complements of the indexing set are taken with respect to $\pazocal{N}$: if $\pazocal{J}\subset\pazocal{N}$, then $\pazocal{J}^\complement\equiv\pazocal{N}\setminus \pazocal{J}$. Given any set $\pazocal{J}\subset\pazocal{N}$ and any $x\in X$, we use $x_{-\pazocal{J}}=x_{\pazocal{J}^\complement}$ and $X_{-\pazocal{J}}=X_{\pazocal{J}^\complement}$. 

\indent If $E$ is a finite set, we use $\Delta(E)$ to denote the probability distributions over $E$. Finally, if $\pazocal{J}\subset \pazocal{N}$ and $E\subset X_{\pazocal{J}}$, then $\mathbbm{1}_{E}:X_{\pazocal{J}}\rightarrow \{0,1\}$ denotes the indicator function that event $E$ has occurred; that is, $\mathbbm{1}_{E}(x_\pazocal{J})=1\Leftrightarrow x_\pazocal{J}\in E$.

\paragraph{Graph-theoretic notation} A directed graph is a pair $(V,E)$ such that $V$ is the (finite) set of nodes and $E\subset V\times V$ is the set of edges. If two nodes, $i$ and $j$, satisfy that $(i,j)\in E$, we write $i\rightarrow j$. The set of \emph{parents} of a node $v\in V$ is denoted $Pa(v)=\{v'\in V: (v^\prime,v)\in E\}$. A node $v\in V$ is a descendant of a node $v'\in V$ whenever a directed path exists from $v'$ to $v$. Formally, if a sequence $(v_1,...,v_T)\in V^T$ exists such that $v_1=v^\prime$, $v_t$ is a parent of $v_{t+1}$ for each $t\in\{1,...,T-1\}$, and $v_T=v$. Analogously, $v^\prime$ is an ancestor of $v$ whenever $v$ is a descendant of $v^\prime$. We denote by $D(v)$  and $ND(v)$ the set of descendants and non-descendants of $v$, respectively. Finally, a directed graph is a directed \emph{acyclic} graph (DAG) if and only if for all $v\in V$, $v$ is not a descendant of $v$. 

\subsection{Expanding the Savage framework with intervention preferences}\label{sec:savage-expansion}

\paragraph{The standard Savage problem} In the standard Savage problem a decision maker (DM) chooses between acts, which are functions mapping states of the world into monetary outcomes. When choosing an act, the DM treats the state of the world as a random variable chosen by \emph{nature}: the DM chooses an act, $f$, nature selects a state of the world, $s$, and the DM receives the payment $f(s)$. Importantly, the DM cannot control nature's choice of the state.

\indent In the context of our paper, the Savage framework is characterized by states $S=\Pi_{i=1}^NX_i\equiv X$, acts $\pazocal{F}=\mathbb{R}^{X}$, and preferences $\succ$ defined over $\pazocal{F}$. We refer to each $i\in\pazocal{N}$ as a \emph{variable}, and we assume each $X_i$ is finite. We assume a finite state space for technical simplicity because causality is unrelated to whether state spaces are finite or infinite. 

\indent Throughout the paper we apply the following simplifying notation. Because $\succ$ is defined on $\mathbb{R}^X$, acts on both side of a $\succ$ comparison should be defined on $\mathbb{R}^X$, but most of the time we care about comparing acts that only differ on a subset of variables. Suppose $f$ and $g$ are acts defined on some subset of variables, $\pazocal{J}\subsetneq \pazocal{N}$, so $f,g\in\mathbb{R}^{X_{\pazocal{J}}}$. To compare $f$ and $g$ using $\succ$ we need to extend $f$ and $g$ to the full state space, $X=(\pazocal{N}\setminus\pazocal{J})\cup\pazocal{J}$. We extend $f$ and $g$ to the full state space multiplying by a function that is constantly $1$ for any realization of $\pazocal{N}\setminus\pazocal{J}$, so we can compare $\mathbb{1}_{X_{\pazocal{N}\setminus\pazocal{J}}}f$ with $\mathbb{1}_{X_{\pazocal{N}\setminus\pazocal{J}}}g$ using $\succ$: these two acts are defined over the entire state space but they only depend non-trivially on $\pazocal{J}$ as indicated by $f$ and $g$. To avoid this cumbersome notation we simplify notation and simply write $f\succ g$ to indicate $\mathbb{1}_{X_{\pazocal{N}\setminus\pazocal{J}}}f \succ \mathbb{1}_{X_{\pazocal{N}\setminus\pazocal{J}}}g$ whenever confusion is unlikely.

\paragraph{Policy Interventions} To model causality we extend the Savage framework by incorporating \emph{policy interventions}. In contrast with the standard Savage model, policy interventions are a tool with which the DM can control the state's realization. Formally, the set of \emph{policy interventions} is the set $\pazocal{P}=\Pi_{i=1}^N(X_i\cup\{\emptyset\})$, interpreted as follows. For a policy $p\in\pazocal{P}$, and a variable $i\in\pazocal{N}$, suppose that $p_i=\emptyset$. Then, this policy leaves variable $i$ unaffected; that is, $i$ is determined by nature. In contrast, suppose that for some other variable, $j\in\pazocal{N}$, $p_j=x_j\in X_j$. Then policy $j$ forces variable $j$ to take the value $x_j$; that is, the value of variable $j$ is directly chosen by the DM and not by nature as it was in the Savage model. We then say that $j$ is an \emph{intervened variable}. Therefore, each policy determines a collection of \emph{interventions} on the state space.

\paragraph{Choice Domain} Our DM chooses a policy intervention and a Savage act over the non-intervened variables. Concretely, let $p\in\pazocal{P}$ be any policy, and $\pazocal{N}(p)=\left\{i\in\pazocal{N}: p_i=\emptyset\right\}$ be the set of non-intervened variables. In contrast with the standard Savage framework, where acts are defined on $\mathbb{R}^X$, our DM chooses over acts defined over the non-intervened variables: $\pazocal{F}(p)\equiv \mathbb{R}^{X_{\pazocal{N}(p)}}$. The DM's primitive choice domain is the set $\left\{(p,f): p\in\pazocal{P}, f\in\pazocal{F}(p)\right\}$. We endow this DM with a preference relation, $\overline{\succ}$, on $\left\{(p,f): p\in\pazocal{P}, f\in\pazocal{F}(p)\right\}$. 

\paragraph{Two types of induced decision problems: conditioning and intervening} Because our DM chooses both interventions and standard Savage acts, there are two types of induced decision problems they may face. 
For example, consider acts $f$ and $g$ in \autoref{fig:conditional}:
\begin{figure}[h!]
\centering
\subfloat[Act $f$]{
\begin{tabular}{c|cc}
&$L< \$ 3M$&$L\geq \$ 3M$\\
\hline
$E\neq$ College, $a\in A$ &\begin{tabular}{c}
\textcolor{black}{$0$}
\end{tabular}&\begin{tabular}{c}
\textcolor{black}{$0$}
\end{tabular}\\
$E=$ College, $a\in A$&\begin{tabular}{c}
\textcolor{black}{$1$}
\end{tabular}&\begin{tabular}{c}
\textcolor{black}{$0$}\end{tabular}
\end{tabular}
}
\subfloat[Act $g$]{\begin{tabular}{c|cc}
&$L < \$ 3M$&$L\geq \$ 3M$\\
\hline
$E\neq$ College, $a\in A$ &\begin{tabular}{c}
$0$
\end{tabular}&\begin{tabular}{c}
$0$
\end{tabular}\\
$E=$ College, $a\in A$&\begin{tabular}{c}
$0$
\end{tabular}&\begin{tabular}{c}
$1$
\end{tabular}
\end{tabular}
}
\caption{Illustrating conditional preferences}\label{fig:conditional}
\end{figure}\newpage

\indent Acts $f$ and $g$ pay out as a function of citizen Kane's lifetime earnings, but they only pay out if nature selects $E=$ college. This is an example of a Savage conditional act (in this case, conditional on education levels), from which we could define the DM's conditional beliefs (in this case, beliefs on lifetime earnings conditional on education levels). Importantly, comparing $f$ and $g$ is done when no variables are intervened (equivalently, after selecting the policy intervention $p=(\emptyset, \emptyset, \emptyset)$).

\indent A different induced decision problem is to first intervene citizen Kane's education level to $E=$ college, and then choose amongst acts over Kane's realized lifetime earnings like $f$ and $g$ above. Formally, the DM chooses policy intervention $=p\equiv (p_A, p_E, p_L)=(\emptyset, College, \emptyset )$ after which they compare $(p,f')$ and $(p,g')$ in \autoref{fig:intervention} using preference $\overline{\succ}$:

\begin{figure}[h!]
\centering
\subfloat[Act $f^\prime$]{
\begin{tabular}{c|cc}
&$L< \$ 3M$&$L\geq \$ 3M$\\
\hline
$a\in A$ &\begin{tabular}{c}
\textcolor{black}{$1$}
\end{tabular}&\begin{tabular}{c}
\textcolor{black}{$0$}
\end{tabular}\\
\end{tabular}\label{fig:actfp}}
\subfloat[Act $g^\prime$]{
\begin{tabular}{c|cc}
&$L < \$ 3M$&$L\geq \$ 3M$\\
\hline
$a\in A$ &\begin{tabular}{c}
$0$
\end{tabular}&\begin{tabular}{c}
$1$
\end{tabular}\\
\end{tabular}\label{fig:actgp}}
\caption{Illustrating intervention independence}\label{fig:intervention}
\end{figure}

\indent This is an example of an intervention preference, from which we define intervention beliefs: for fixed a policy interventions, we compare different acts on the non-intervened variables. 

\indent Comparing intervention beliefs and conditional beliefs is at the heart of Bayesian causality. We conclude this section by formally defining these concepts formally and comparing their properties.

\paragraph{Conditional Preferences and Conditional Beliefs} We now formally define \emph{Savage conditional acts}, from which \emph{conditional preferences} and \emph{conditional beliefs} are defined. These definitions are adaptations of classical definitions in \cite{savage1972foundations} to our setup. Consider a variable, $i$, an act over $X_i$, $f\in\mathbb{R}^{X_i}$, a set of variables, $\pazocal{J}\subset\pazocal{N}_{-\{i\}}$, and a realization $x_\pazocal{J}\in X_\pazocal{J}$. The act $\mathbb{1}_{x_\pazocal{J}}f\in\mathbb{R}^{X_{\pazocal{J}\cup\{i\}}}$ is the act that pays according to $f$ if $x_\pazocal{J}$ is realized but pays \$0 otherwise; we refer to this as a Savage conditional act. From here, we define the DM's preference on $\mathbb{R}^{X_i}$ conditional on $x_{\pazocal{J}}$ being realized: this is the restriction of $\succ$ to the sub-space of acts $\{\mathbb{1}_{x_\pazocal{J}}f:f\in\mathbb{R}^{X_i} \}\subset\mathbb{R}^{X_{\pazocal{J}\cup\{i\}}}$. We refer to this as a conditional preference.

\indent Conditional acts and conditional preferences are key in defining conditional beliefs (and, later on, conditional probability). Given a variable $i\in\pazocal{N}$, and two events $E,F\subset X_{i}$, we say that $E$ is (unconditionally) more likely than $F$  if the DM prefers an act that pays \$1 in $E$ and \$0 otherwise to an act that pays \$1 in $F$ and \$0 otherwise. Formally, if $\mathbb{1}_E\succ \mathbb{1}_F$. Analogously, for a set of variables, $\pazocal{J}\subset\pazocal{N}_{-\{i\}}$, and an event $J\subset X_{\pazocal{J}}$ we can say that $E$ is more likely than $F$ conditional on $J$ if the DM conditionally prefers the act $\mathbb{1}_E$ to the act $\mathbb{1}_F$. Formally, if $\mathbb{1}_J\mathbb{1}_E\succ \mathbb{1}_J\mathbb{1}_F$. The DM's conditional beliefs are thus expressions of the DM's conditional preferences, obtained by pre-multiplying the relevant acts by the indicator function of the relevant conditioning event. 

\paragraph{Intervention Preferences and Intervention Beliefs} As with conditional preferences and conditional beliefs, interventions generate the analogous notions of \emph{intervention preferences} from which \emph{intervention beliefs} are derived.

\indent  Given $\overline{\succ}$, each $p$ induces an \emph{intervention preference} on $\pazocal{F}(p)$: for each $p\in\pazocal{P}$ and each $f,g\in\pazocal{F}(p)$, we say that $f\succ_pg$ if, and only if, $(p,f)\overline{\succ} (p,g)$. Because our axioms focus on the DM's intervention preferences, it is convenient to express intervention preferences explicitly in terms of the values at which the DM intervenes on the variables. For each policy $p\in\pazocal{P}$, if $p_{-\pazocal{N}(p)}=x_{-\pazocal{N}(p)}$, we use $\succ_{x_{-\pazocal{N}(p)}}$ to denote $\succ_{p_{-\pazocal{N}(p)}}$. In particular, the DM is allowed to choose a policy that does not intervene any variable; that is, $p=(\emptyset,...,\emptyset)$. Then, $\pazocal{F}(\emptyset,...,\emptyset)=\mathbb{R}^X$, and $\succ_{(\emptyset,...,\emptyset)}$ embed the standard Savage framework in our expanded framework. To simplify notation, we write $\succ_{(\emptyset,...,\emptyset)}\equiv \succ$.

\indent From intervention preferences, we obtain \emph{intervention beliefs}. For each policy $p\in\pazocal{P}$, let $E,F\subset X^{\pazocal{N}(p)}$ be two events in the space of non-intervened variables. We say $E$ is more likely than $F$ after intervention $p$ if an act that pays \$$1$ when $E$ happens and \$$0$ otherwise is $\succ_p$-preferred than an act that pays \$$1$ when $F$ happens and \$$0$ otherwise. Formally, if $\mathbb{1}_{E}\succ_p \mathbb{1}_F$. Intervention beliefs are analogous to conditional beliefs except that the DM is intervening a subset of variables, rather than conditioning acts upon nature selecting a specific realization for those variables.

\indent For each $p\in\pazocal{P}$, we say that $\succ_p$ has a \emph{probabilistic representation} if there is a probability distribution $\mu_p$ on $X_{\pazocal{N}(p)}$ such that for all $E,F\subset X_{\pazocal{N}(p)}$, $\mu_p(E)>\mu_p(F)$ if, and only if, $\mathbbm{1}_E\succ_p\mathbbm{1}_{F}$. When such a representation exists, we say that $\mu_{p}$ is an \emph{intervention probability}. 

\indent Intervention preferences resemble Savage conditional preferences but have important differences. Savage conditional preferences capture betting behavior conditional on the DM observing that a certain event was realized, whereas intervention preferences capture betting behavior after intervening the relevant variables. 

\section{Definition of causality}\label{sec:defcausal}
 In this section, we introduce the definition of Bayesian causality, which formalizes the intuitive definition given in the introduction. 

 Bayesian causality (\autoref{def:causality}) is defined as a lack of intervention independence, so we begin by defining intervention independence. Informally, intervention independence means the value of information contained in an intervention is zero. Consider a set of variables, $\pazocal{K}$, and two other variables, $i,j\notin\pazocal{K}$. Suppose the DM already intervened the variables in $\pazocal{K}$ and now has to choose between two acts, $f,g\in\mathbb{R}^{X_i}$. Because $f$ and $g$ pay out as a function of $i$'s realization, the DM values information about $i$. In addition to the intervention of $\pazocal{K}$, the DM is offered the chance to also intervene variable $j$ at a small fee $\varepsilon>0$. If the choice of acts over $i$ is insensitive to further interventions on $j$, then intervening $j$ is not worth any $\varepsilon>0$ fee. We then say that $i$ is $\pazocal{K}$-intervention independent of $j$.

\begin{Definition}[Intervention independence]\label{def:independence}
We say that variable $i$ is $\pazocal{K}$-\emph{intervention independent of variable $j$} if $\{i,j\}\cap\pazocal{K}=\emptyset$ and for all $f,g\in\mathbb{R}^{X_i}$ and $x_j\in X_j$,
\begin{eqnarray*}
f\succ_{x_{j},x_{\pazocal{K}}}g &\Leftrightarrow& f\succ_{x_{\pazocal{K}}} g,
\end{eqnarray*}
\end{Definition}

\begin{Remark}
\autoref{def:independence} also implies the following condition holds for all $x_j,x_j^\prime\in X_j$
\begin{eqnarray*}
f\succ_{x_{j},x_{\pazocal{K}}} g &\Leftrightarrow&  f\succ_{x_{j}',x_{\pazocal{K}}}g. 
\end{eqnarray*}
Indeed, for any two values of $j$, $x_j$ and $x'_{j}$, \autoref{def:independence} implies
\begin{eqnarray*}
f\succ_{x_{j},x_{\pazocal{K}}} g \Leftrightarrow f\succ_{x_{\pazocal{K}}} g \Leftrightarrow f\succ_{x'_{j},x_{\pazocal{K}}} g.
\end{eqnarray*}
The above condition states that if the DM sees no value in intervening $j$ after $\pazocal{K}$ is intervened, then the DM's intervention preferences, $\succ_{x_{\pazocal{K} \cup \{j\}} }$, should not change with the value of $x_j$.
\end{Remark}

\indent Having defined intervention independence, we can formally define Bayesian causality:

\begin{Definition}[Bayesian causality]\label{def:causality}
For all $i,j\in\pazocal{N}$, we say that variable $j$ \emph{Bayesian causes} variable $i$ if $i$ is not $\{i,j\}^\complement$-independent of $j$. 
\end{Definition}

 To illustrate the role of  \autoref{def:independence} in defining Bayesian causality, consider a DM (say, Blake) with the following beliefs: $E$ducation is intrinstically informative about $A$bility, and ability is intrinsically informative about $L$ifetime earnings, but education only conveys information about lifetime earnings through the ability variable. If $e,e'\in E$ are two education levels, and $f$ and $g\in \mathbb{R}^{X_L}$ are two acts on lifetime earnings, both $f\succ_{e} g$ and $g\succ_{e'} f$ could be true for Blake. This reversal indicates that $E$ and $L$ are not $\{\emptyset\}$-independent, which is intuitive: interventions on $E$ affect beliefs about $A$, and beliefs about $A$ affect beliefs about $L$. However, this is an effect of $E$ on $L$ that is mediated through $A$. As such, we would not want to use this as a basis to claim that $E$ \emph{causes} $L$--at least not if we want to define causality as the unmediated effect that one variable has on the beliefs about another variable. To define ``$E$ causes $L$'' we instead regard intervention preferences $\succ_{(a, e)}$ as a function of $e$ for each fixed $a\in A$. In other words, we want to ask if $E$ and $L$ are $\{A\}$-independent. This motivates the formal definition of Bayesian causality above.

\paragraph{Bayesian causality notation} We now collect basic notation and nomenclature related to Bayesian causality. We let $Ca(i)=\{j\in\pazocal{N}:\text{$j$ is a Bayesian cause of $i$}\}$ denote the \emph{Bayesian causes of $i$}, and $Co(i)=\{j: \text{$i$ is a Bayesian cause of $j$}\}$ denote the set of \emph{Bayesian consequences of $i$}. Furthermore, we say that $j$ is an indirect Bayesian cause of $i$ if there is a sequence $j_0,...,j_T$ such that, for all $t\in\{0,...,T-1\}$, $j_{t}$ causes $j_{t+1}$, $j_0=j$ and $j_{T}=i$. In Blake's example above, $E$ is not a Bayesian cause of $L$ but it is an indirect Bayesian cause of $L$. We denote the set of indirect Bayesian causes of a variable $i$ by $ICa(i)$, and (analogously) the set of indirect Bayesian consequences by $ICo(i)$. We also use the notation above when referencing subsets of variables: if $\pazocal{J}\subset\pazocal{N}$ is a set of variables, then $Ca(\pazocal{J})=\cup_{j\in\pazocal{J}} Ca(j)$, and analogously for consequences, indirect causes, and indirect consequences. Finally, if a variable $i$ is such that $Ca(i)=\emptyset$, we say that $i$ is an \emph{exogenous primitive}. Indeed, when a DM forms a causal model of the world, the set of primitives of such a model is precisely the set of variables that are not caused by any other variable in the model. 

\section{Representations}\label{sec:representations}
\indent We present four representations for four objects, which together represent the DM's causal model. First, we present the Subjective Expected Utility (SEU) representation of the DM's intervention preferences. This representation summarizes the DM's probabilistic view of the world. Second, given any probability distribution, $\mu\in\Delta(S)$, we present a DAG representation of $\mu$'s conditional independence structure. This representation comes from the conditional independence literature \citep{lauritzen1990independence}. Third, we use this definition to define a DAG representation of the DM's primitive preference, $\overline{\succ}$. Finally, we define the DM's \emph{causal graph} as a graph where arrows point from Bayesian cause to Bayesian consequence.

\paragraph{Expected utility representation of intervention preferences} In this representation, intervention beliefs are represented via a collection of probability distributions, each indexed by the intervention that generates them. 

\indent We say that an intervention preference, $\succ_p$, is a (monotone) subjective expected utility preference if there exists a unique probability distribution, $\mu_p\in\Delta(X_{\pazocal{N}(p)})$, and a (monotone increasing) function $u_p:\mathbb{R}\rightarrow\mathbb{R}$ such that for all acts, $f,g\in \mathbb{R}^{X_{\pazocal{N}(p)}}$,
\begin{eqnarray}\label{eq:EU}
f\succ_{p} g \Leftrightarrow \sum_{x_{\pazocal{N}(p)}\in X_{\pazocal{N}}(p)}u_{p}(f(x_{\pazocal{N}(p)}))\mu_{p}(x_{\pazocal{N}(p)})>\sum_{x_{\pazocal{N}(p)}\in X_{\pazocal{N}}(p)}u_{p}(g(x_{\pazocal{N}(p)}))\mu_{p}(x_{\pazocal{N}(p)}).
\end{eqnarray}

\indent There are many axiomatizations of monotone expected utility preferences that fit the framework of our model, such as \cite{Gul91},  Theorem $3$ in \cite{fishburn1967preference}, and Theorem $3$ in  \cite{karni2006subjective}, among others. Because axiomatizing SEU is not the focus of this paper, we let the readers select their preferred axiomatization. 

\indent We assume the intervention preferences satisfy the following:

\begin{Assumption}\label{ass:ass1}
For each $p,p'\in\pazocal{P}$, the following are true.
\begin{enumerate}[label=(\roman*)]
\item The intervention preference $\succ_{p}$ is a monotone subjective expected utility preference.
\item There are no null states: for all $x\in X$, $\mathbbm{1}_x\succ_p \mathbbm{1}_{X}0$.
\item Policies do not affect preferences: $(\forall x,y\in\mathbb{R})$, $\mathbbm{1}_{X_{\pazocal{N}(p)}}x\succ_p\mathbbm{1}_{X_{\pazocal{N}(p)}}y \Leftrightarrow \mathbbm{1}_{X_{\pazocal{N}(p')}}x\succ_{p'}\mathbbm{1}_{X_{\pazocal{N}(p')}}y$.
\end{enumerate}
\end{Assumption}

\autoref{ass:ass1} guarantees three properties of the DM's probabilistic model. First, each of the DM's intervention preferences admits an SEU representation. Second, for any intervention, each state has positive probability: because the act that pays \$$1$ if $x\in X$ is realized is strictly preferred to the act that pays $0$ under all states, then each $x\in X$ must be realized with positive probability. Lastly, Bernoulli indices are independent of the policy choice: $u_p=u_{p'}$ for all $p,p'\in\pazocal{P}$. This is a natural assumption because Bernoulli indices represent the DM's preference over monetary outcomes, which are naturally independent from policy choice. We call these three properties an assumption, rather than an axiom, because this part of the model deals only with the DM's probabilistic model, not their causal model. 

\paragraph{DAG representation of conditional independence} \cite{lauritzen1990independence} provides a definition for when a DAG represents a given probability distribution, say $\mu\in\Delta(\Pi_{i\in\pazocal{N}}X_i)$. The goal is to graphically represent the conditional independence structure of $\mu$, not to make causal statements. 

\indent We first motivate the DAG representation of conditional independence with the example in \autoref{fig:representation-graph} below:

\begin{figure}[H]
\centering
\scalebox{1}{%
\begin{tikzpicture}[thick]
 nodes
\draw (0,0) node {$a$};
\draw (0,4) node {$b$};
\draw (2,0) node {$j$};
\draw (2,4) node {$w$};
\draw (4,2) node {$i$};
\draw (6,2) node {$k$};
\draw (8,2) node {$z$};
\draw[black, ->] (0,0.5)--(0,3.5);
\draw[black, ->] (0.5,0)--(1.5,0);
\draw[black, ->] (2.5,0)--(3.5,1.5);
\draw[black, ->] (1.5,4)--(0.5,4);
\draw[black, ->] (2.5,4)--(3.5,2.5);
\draw[black, ->] (4.5,2)--(5.5,2);
\draw[black, ->] (6.5,2)--(7.5,2);

\end{tikzpicture}}
\caption{A DAG representing the distribution $\mu(a,b,w,j,i,k,z)=\mu(a)\mu(w)\mu(b\vert w,a)\mu(j\vert a)\mu(i\vert w,j)\mu(k\vert i)\mu(z\vert k)$.}\label{fig:representation-graph}
\end{figure}

 The DAG above represents a distribution, $\mu\in\Delta(\{a,b,w,j,i,z\})$, that can be factorized as $\mu(a,b,w,j,i,k,z)=\mu(a)\mu(w)\mu(b\vert w,a)\mu(j\vert a)\mu(i\vert w,j)\mu(k\vert i)\mu(z\vert k)$. The nodes in the DAG correspond to the variables, and an arrow from a node $v$ to a node $v'$ encodes that $v$ contains information about $v'$ that is not contained in any other variable, $v''$. Thus, the parents of a variable $v$, $Pa(v)$, are the variables that provide \emph{unmediated information} about $v$. For instance, \autoref{fig:representation-graph} represents that $w$ and $j$ contain unmediated information about $i$, and thus that $i$ is never independent of $\{w,j\}$. Similarly, $i$ is never independent of its direct descendant, $k$, because $i$ carries unmediated information about $k$. Now, consider a variable that is an ancestor of $i$; for example, $a$. Clearly, $a$ and $i$ are not independent: $a$ provides unmediated information about $j$, which in turn provides information about $i$. However, any information that $a$ has about $i$ is implicitly encoded in $j$. Indeed, if $a$ carries unmediated information about $i$, there should be an arrow $a\rightarrow i$, but such an arrow is absent. Similarly, $b$ provides information about $i$: $b$ is informative about $\{a, w\}$, both of which are informative about $i$. However, any information that $b$ has about $i$ is encoded in $\{j,w\}$. This implies that once we condition on the parents of $i$ (in this case, $\{w,j\}$), all non-descendants of $i$ are conditionally independent of $i$, because any information about $i$ that a non-descendant might contain is mediated via $\{w,j\}$. 

\indent More generally, let $\mu\in\Delta(\Pi_{i\in\pazocal{N}}X_i)$ be a probability distribution, and let $G=(\{1,...,N\}, E)$ be a DAG. The chain rule implies the following:
\begin{eqnarray}\label{eq:chainrule}
 (\forall x\in \Pi_{i\in\pazocal{N}}X_i),\text{ } \mu(x)=\Pi_{i=1}^N\mu(x_i\vert ND(i)),
\end{eqnarray}
where recall that $ND(i)$ is the set of non-descendants of variable $i$. Because arrows are drawn so that any information a non-descendant of $i$ contains about $i$ is mediated through $i$'s parents, the terms $\mu(x_i\vert ND(i))$ in \autoref{eq:chainrule} simplify to $\mu(x_i\vert Pa(i))$. This observation motivates \autoref{def:representation0} below.

\begin{Definition}\label{def:representation0}
Let $\mu\in\Delta(\Pi_{i\in\pazocal{N}}X_i)$. The DAG $(\{1,...,N\}, E)$ represents $\mu$ if, and only if, the following hold: $(\forall x\in \Pi_{i\in\pazocal{N}}X_i)$,
\begin{eqnarray*} 
\mu(x)&=&\Pi_{i=1}^N\mu(x_i\vert Pa(i))
\\ 
(\forall (\pazocal{T}_i)_{i\in\pazocal{N}})(\pazocal{T}_i\subset Pa(i)),\text{ if } \mu(x)&=&\Pi_{i=1}^N\mu(x_i\vert \pazocal{T}_i)\Rightarrow (\forall i\in\pazocal{N}), \pazocal{T}_i=Pa(i).
\end{eqnarray*}
\end{Definition}
\autoref{def:representation0} states two things. First, a DAG represents a probability distribution if, and only if, the DAG summarizes the conditional independence properties of $\mu$ in the sense discussed previously. Second, the set of parents is the smallest set that allows for such a decomposition. Indeed, consider a set of nodes $V=\{A,E,L\}$ and a probability distribution $\mu(x_A,x_E,x_L)=\mu(x_A)\mu(x_E)\mu(x_L)$. Since all variables are probabilistically  independent, both DAGs in \autoref{fig:minimality} represent this $\mu$. Indeed, both $\mu(x_A,x_E,x_L)=\mu(x_A)\mu(x_E \vert x_A)\mu(x_L\vert x_E)$ and $\mu(x_A,x_E,x_L)=\mu(x_A)\mu(x_E)\mu(x_L)$ are true statements. However, the first representation includes irrelevant arrows, which the minimality requirement prevents.

\begin{figure}[H]
\centering
\scalebox{1}{%
\begin{tikzpicture}[thick]
 nodes
\draw (-2,0) node {$A$};
\draw (1,0) node {$E$};
\draw (4,0) node {$L$};
\draw (-2,-2) node {$A$};
\draw (1,-2) node {$E$};
\draw (4,-2) node {$L$};
 arrows
\draw[black, ->] (-1.5,0)--(0.5,0);
\draw[black, ->] (1.5,0)--(3.5,0);
\end{tikzpicture}}
\caption{Both DAGs above represent the same probability distribution, $\mu(x_A,x_E,x_L)=\mu(x_A)\mu(x_E)\mu(x_L)$, but the top one includes irrelevant arrows.}\label{fig:minimality}
\end{figure}

\indent Importantly, graphs that represent conditional independence have no intrinsic causal interpretation. To see this, consider \autoref{fig:reverse} below. Both DAGs in the figure represent the same probability distribution: $\mu(A,E,L)=\mu(E)\mu(A\vert E)\mu(L\vert A)$. If arrows had causal meaning, then education would both be cause of ability (because $E\rightarrow A$ in \autoref{fig:g}) and $E$ would not be a cause of ability ($E\nrightarrow A$ in \autoref{fig:ghat}), which is a contradiction.

\begin{figure}[H]
\centering
\subfloat[Graph $G_{\mu}$]{\begin{tikzpicture}[thick,scale=1.5]
\draw (1,3) node {$E$};
\draw (1,1) node {$A$};
\draw (3,3) node {$L$};

\draw[black, ->] (1,1.25) -> (1,2.75);
\draw[black, ->] (1.25,1.21) -> (2.75,2.75);
\end{tikzpicture}\label{fig:g}}\hspace{1cm}
\subfloat[Graph $\hat{G}_{\mu}$.]
{\begin{tikzpicture}[scale=1.5,thick]
\draw (1,3) node {$E$};
\draw (1,1) node {$A$};
\draw (3,3) node {$L$};

\draw[black, ->] (1,2.75) -> (1,1.25);
\draw[black, ->] (1.25,1.21) -> (2.75,2.75);
\end{tikzpicture}\label{fig:ghat}}
\caption{Two directed acyclic graphs that represent $\mu$, $\mu(A,E,L)=\mu(A)\mu(E\vert A)\mu(L\vert A)$.}\label{fig:reverse}
\end{figure}

\paragraph{DAG representation of the DM's preference, $\overline{\succ}$} Using \autoref{def:representation0}, we can define when a graph represents a standard Savage preference. Suppose that $\succ$ is the DM's Savage preference defined on $\mathbbm{R}^{X}$. Under \autoref{ass:ass1}, there is a well-defined probabilistic representation of the DM's subjective beliefs encoded in $\succ$, $\mu\in\Delta(S)$. We then say that graph $G$ represents $\succ$ if $G$ represents $\mu$ in the sense of  \autoref{def:representation0}.

\indent However, \autoref{def:representation0} is insufficient to define when a graph represents a preference $\overline{\succ}$. Indeed, a preference $\overline{\succ}$ is associated with the collection of intervention preferences $\{\succ_p:p\in\pazocal{P}\}$. As such, $\overline{\succ}$ is associated with a family of beliefs, rather than a single belief, as in Savage's model. 

\indent To define when a DAG represents preferences $\overline{\succ}$, we first define the \emph{truncation} of a DAG. Let $G=(V,E)$ be a DAG, and let $W\subsetneq V$. The $W$-truncated DAG, $G_{W}$, is the DAG obtained by eliminating all nodes in $W$, together with their incoming and outgoing arrows. For example, \autoref{fig:representation-truncation} shows a DAG, $G$, on the left, and the truncation, $G_{E}$, on the right. After variables in $W$ are intervened on, they no longer form part of the DM's probabilistic model; they are now deterministic objects that are probabilistically  uninformative about the value of their parents. Thus, we exclude these variables from the corresponding DAG. In the context of our running example, if Alex observes that Mr. Kane obtained a college degree, his education is no longer random, but Alex can still make inferences about Mr. Kane's intellectual ability. Thus, education remains a legitimate element of Alex's probabilistic model. However, if Mr. Kane's education is intervened to ``college degree'', then his education level is no longer random and, furthermore, is uninformative about his ability level. Thus, we exclude education from Alex's post-intervention model. For example, the DAG in \autoref{fig:full} is a representation of $\succ$--Alex's preferences when no variables are intervened--and the  DAG in \autoref{fig:trunc} is a representation of $\succ_{x_E}$--Alex's preferences when education is intervened to any level $x_E$. From $G$ and $G_E$ we can deduce that Alex's beliefs satisfy $\mu(x_A,x_E,x_L)=\mu(x_A)\mu(x_E\vert x_A)\mu(x_L\vert x_A)$--earnings and education are independent conditional on ability--and $\mu_{x_E}(x_A,x_L)\neq\mu_{x_E}(x_A)\mu_{x_E}(x_L)$--ability is not independent of earnings after education is intervened.

\begin{figure}[H]
\centering
\subfloat[Full model]{
\scalebox{1}{%
\begin{tikzpicture}[thick]
\draw (-2,-2) node[below=3pt] {$A$};
\draw (-2,2) node[above=3pt] {$E$};
\draw (2,0) node[right=3pt] {$L$};
\draw[black, ->] (-2,-2) -- (-2,2);
\draw[black, ->] (-1.7,-2) -- (1.7,-0.3);
\end{tikzpicture}\label{fig:full}}
}\hspace{1cm}
\subfloat[After intervening $E$]{
\scalebox{1}{%
\begin{tikzpicture}[thick]
\draw (-2,-2) node[below=3pt] {$A$};
\draw (2,0) node[right=3pt] {$L$};
\draw[black, ->] (-1.7,-2) -- (1.7,-0.3);
\end{tikzpicture}\label{fig:trunc}}
}
\caption{Before and after intervening $E$. The left panel shows the full model, in which $E$ is informative about $L$ because knowing $E$ is informative of its cause, $A$. The right panel shows the model after intervening $E$, so that it is no longer part of the probabilistic model as $E$ is uninformative about its causes.}\label{fig:representation-truncation}
\end{figure}

 We can now define when a DAG, $G$, represents a preference $\overline{\succ}$. We  say that a graph represents a preference if, for each intervention, the appropriately truncated subgraph represents the corresponding intervention preference. 

\begin{Definition}[Conditional independence DAG of $\overline{\succ}$]\label{def:representation2}
Let $G=(\pazocal{N},E)$ be a DAG and $\overline{\succ}$ be a preference. Assume that for each $\pazocal{K}\subset\pazocal{N}$ and each $x_{\pazocal{K}}\in X_\pazocal{K}$, $\succ_{x_{\pazocal{K}}}$ has a well-defined probabilistic representation; let $\mu_{x_{\pazocal{K}}}$ be the corresponding representation. We say that $G$ \emph{represents} $\overline{\succ}$ if the following are true for each $\pazocal{K}\subset\pazocal{N}$ and each $x\in X$:
\begin{enumerate}[label=(\roman*)]
\item $G_{\pazocal{K}}$ represents $\mu_{x_\pazocal{K}}$,
\item If $(i,j)\in E$ then $\mu_{x_{-\{i,j\}}}(x_j\vert x_i)= \mu_{x_{-\{j\}}}(x_j)$.
\end{enumerate}
\end{Definition}

 Note that nothing in this section so far is related to causality. Indeed, the statement that a graph represents a probability distribution is purely a statement about probabilistic independence. As such, the representation of a probability by a DAG is a statement about conditional independence, not causation. 

\paragraph{Causal graphs} We conclude this section by defining the causal graph associated with a preference, $\overline{\succ}$. Given $\overline{\succ}$, draw a graph by letting the set of nodes be the set of variables, and the set of arrows be defined by the causal sets: $j\rightarrow i\Leftrightarrow j\in Ca(i)$. This graph is well defined because $Ca(i)$ is well defined for each $i\in\pazocal{N}$ via \autoref{def:causality}. We denote such a graph by $G(\overline{\succ})$. 

\begin{Definition}[Causal DAG of $\overline{\succ}$]\label{def:representation3}
Let $\overline{\succ}$ be a DM's preference, and, for each $i\in\pazocal{N}$. The \emph{causal graph of $\overline{\succ}$} is the graph $(V,E)=(\pazocal{N}, \{ (j,i)\: : \: i\in\pazocal{N},\: j\in Ca(i)\})$. Henceforth, we denote the causal graph of $\overline{\succ}$ as $G(\overline{\succ})$.
\end{Definition}
%
%

\section{Axioms}\label{sec:axioms}


 Because graphs that represent the conditional independence structure of a distribution may be unrelated to graphs that represent any given definition of ``causality'', we need a way to connect causality and conditional independence. \autoref{th:theorem1} provides this connection. The theorem states that Axioms \ref{ax:axiom3} through \ref{ax:axiom2}, presented in this section, are necessary and sufficient for a graph, $G$, to simultaneously represent the conditional independence structure in the DM's beliefs (\autoref{def:representation2}) while \emph{also} being the DM's causal graph (\autoref{def:representation3}). Furthermore, $G(\overline{\succ})$ is the only graph for which this holds. In this section, we present the axioms, and in the next section, we formally state \autoref{th:theorem1}.

\indent Our axioms connect conditional independence and intervention independence, so we need simple notation that combines these two notions of independence. We refer to this as \emph{conditional independence after an intervention}. We first present the formal definition and notation, we then motivate how conditional independence operates in the space of preferences.

\begin{Definition}\label{def:independence-notation}
Let $i\in\pazocal{N}$ and let $\pazocal{J},\pazocal{K},\pazocal{H}\subset\pazocal{N}$ be disjoint sets such that $i\notin \pazocal{J}\cup\pazocal{K}\cup\pazocal{H}$. We say that $i$ is independent of $\pazocal{J}$ conditional on $\pazocal{K}$ after intervening on $\pazocal{H}$ if the following is true for all $x_\pazocal{J}\in X_\pazocal{J}$, $x_{\pazocal{K}}\in X_\pazocal{K}$, $x_\pazocal{H}\in X_{\pazocal{H}}$, and all $f,g\in \mathbb{R}^{X_i}$:
\begin{eqnarray}\label{eq:indep-not}
\mathbbm{1}_{x_{\pazocal{K}}}f\succ_{x_{\pazocal{H}}} \mathbbm{1}_{x_{\pazocal{K}}}g \Leftrightarrow \mathbbm{1}_{x_{\pazocal{J}}}\mathbbm{1}_{x_{\pazocal{K}}}f\succ_{x_{\pazocal{H}}} \mathbbm{1}_{x_{\pazocal{J}}}\mathbbm{1}_{x_{\pazocal{K}}}g.
\end{eqnarray}
When the above holds, we write
\begin{eqnarray*}
i \perp_{\pazocal{H}} \pazocal{J} \vert \pazocal{K}.
\end{eqnarray*}
In the case in which $\pazocal{J}$ is a singleton, $\pazocal{J}=\{j\}$, we simply write $i \perp_{\pazocal{H}} j \vert \pazocal{K}$.
\end{Definition}

 Conditional independence after an intervention is tied to the value of information: if $i$ and $j$ are independent, a DM that wants to predict the value of $i$ would never pay for information about $j$. Imagine a DM intervened some variables $\pazocal{H}$ to a specific level, $x_\pazocal{H}$. For instance, the DM might have carried out a controlled experiment, or this could simply be a thought experiment. Furthermore, the DM observes a specific realization of some other variables, $x_{\pazocal{K}}\in X_\pazocal{K}$. If the DM has to choose between $f\in\mathbb{R}^{X_i}$ and $g\in\mathbb{R}^{X_i}$, they would compare $\mathbbm{1}_{x_{\pazocal{K}}}f\in\mathbb{R}^{X_{\pazocal{K}}\times X_i}$ with $\mathbbm{1}_{x_{\pazocal{K}}}g\in\mathbb{R}^{X_{\pazocal{K}}\times X_i}$ using preferences $\succ_{x_{\pazocal{H}}}$. This is the left hand side of \autoref{eq:indep-not} in \autoref{def:independence-notation}. To aid the DM's decision, someone offers to reveal the DM the value of the variables in $\pazocal{J}$ for a fee $\varepsilon>0$. Is there an $\varepsilon$ small enough that the DM would purchase this information? If the DM bought this information about $\pazocal{J}$, then their problem becomes to compare $\mathbbm{1}_{x_{\pazocal{K}}}\mathbbm{1}_{x_{\pazocal{J}}}f$ with $\mathbbm{1}_{x_{\pazocal{K}}}\mathbbm{1}_{x_{\pazocal{J}}}g$ using preferences $\succ_{x_{\pazocal{H}}}$. This is the right hand side of \autoref{eq:indep-not} in \autoref{def:independence-notation}. If the DM's choice is the same in both situations (as in \autoref{def:independence-notation}), the information is useless. Thus, the DM would not accept any price $\varepsilon>0$. This is what conditional independence after an intervention means.

\paragraph{Axioms} Jointly, Axioms \ref{ax:axiom3} through \ref{ax:axiom4} specify that for each variable $i$, only two sets of variables are intrinsically informative about $i$: the causes of $i$, and the consequences of $i$. \autoref{ax:axiom3} states that $Ca(i)\cup Co(i)$ are always informative about $i$, and Axioms \ref{ax:secret} and \ref{ax:axiom4} specify that any information other variables provide about $i$ is mediated by $i$'s causes or consequences. Consequently, the information other variables provide about $i$ is a garbling of the information already contained in $Ca(i)\cup Co(i)$: a decision maker who wants to predict the value of $i$, and already knows the values of $Ca(i)\cup Co(i)$, should not be willing to pay any price to observe the value of other variables.

\begin{Axiom}[Causes are jointly informative]\label{ax:axiom3}
For all $i\in \pazocal{N}$, $\pazocal{J}\subset \pazocal{N}_{-i}$ and $\pazocal{H}\subset\pazocal{N}_{-i}$ that is disjoint from $\pazocal{J}$, if $i$ and $\pazocal{J}$ are independent given $Ca(i)\setminus\pazocal{J}$ after intervening $\pazocal{H}$, then $\pazocal{J}$ is disjoint from the causes of $i$. Formally, 
\begin{eqnarray*}
i\text{ }\perp_\pazocal{H}\text{ } \pazocal{J}\text{ } \vert\text{ } (Ca(i)\setminus\pazocal{J}) \Rightarrow Ca(i)\cap\pazocal{J}=\emptyset.
\end{eqnarray*}

\end{Axiom}

\autoref{ax:axiom3} states that the full set of causes of a variable $i$ is informative about $i$. That is, even if the DM knew information about some, but not all, causes of $i$, the remaining causes are still informative about $i$. In contrast to what \autoref{ax:axiom3} states, suppose that we extract from $Ca(i)$ a proper subset, $\pazocal{J}$, and as a result the DM judges that $i$ and $\pazocal{J}$ are independent conditional on $Ca(i)\setminus\pazocal{J}$. Then, a DM that wants to predict $i$ and knows the realization of the (strict) subset $Ca(i)\setminus\pazocal{J}$ would never pay for information on the remaining causes, $\pazocal{J}$. This implies that only a strict subset of the causes are informative, because $Ca(i)\setminus\pazocal{J}$ provides the same information about $i$ as $Ca(i)$.  \autoref{ax:axiom3} rules this out by requiring that if the DM judges that $i$ and $\pazocal{J}$ are independent after conditioning on $Ca(i)\setminus\pazocal{J}$ and intervening on $\pazocal{H},$ then $\pazocal{J}$ contains no causes of $i$.

\autoref{ax:axiom3} above states that the causes of a variable, $i$, are \emph{jointly} informative about $i$, but it does not state that each cause, $j\in Ca(i)$, is \emph{individually} informative about $i$. To understand the difference, consider the following example three-variable example, $\pazocal{N}=\{j_0,j_1,i\}$. Variables $j_0$ and $j_1$ are independent of each other and take values $0$ or $1$ with probability $\frac{1}{2}$. Suppose that $x_i=1$ if $x_{j_1}=x_{j_0}$ and $x_i=0$ otherwise. It is immediate to see that $i$ is independent of both $j_0$ and $j_1$, but $i$ is not independent of the set $\{j_0,j_1\}$: indeed, while information about  $j_0$ or $j_1$ individually is not valuable, as it will not change the DM's beliefs about $i$, once the DM has information about $j_0$, then information about $j_1$ becomes valuable. The variables $\{i,j_0,j_1\}$ are an example of an encryption key: knowing just one piece of the encryption key is useless, but knowing both is valuable. 

\autoref{ax:secret} below relates the example in the previous paragraph to the DM's causal model. It says that if a series of variables, $j_0,...., j_n$, are neither causes nor indirect consequences of $i$, and they are individually independent of $i$ conditional on $i$'s causes, then they are also jointly independent of $i$ conditional on $i$'s causes. In other words, any encrypted information that a set of variables have about $i$ must be factored in through the causes (or consequences) of $i$.

\begin{Axiom}[Causes decode encryption]\label{ax:secret}
For all $i\in\pazocal{N}$, all $\pazocal{J}=\{j_0,...,j_N\}\subset\pazocal{N}_{-i}$, and all $ \pazocal{K}\subset\pazocal{N}\setminus\left(\{i\}\cup\pazocal{J}\right)$, if $i\notin ICa(\pazocal{J})$ and $\pazocal{J}\subset Ca(i)^\complement$, then
\begin{eqnarray*}
i\perp_{\pazocal{K}} j \vert Ca(i)\setminus \pazocal{K} \text{ for all $j\in\pazocal{J}$ } \Rightarrow i\perp_{\pazocal{K}} \pazocal{J} \vert Ca(i)\setminus \pazocal{K}.
\end{eqnarray*}
\end{Axiom}

\autoref{ax:axiom4} complements Axioms \ref{ax:axiom3} and \ref{ax:secret}, and all three jointly imply that the set of causes of $i$ is informative and a sufficient statistic for $i$. While  \autoref{ax:axiom3} describes the conditional independence properties of variables that cause each other, \autoref{ax:axiom4} describes the independence properties of variables that do \emph{not} cause each other. Furthermore, while \autoref{ax:secret} shows when we can extrapolate joint conditional independence from individual conditional independence, \autoref{ax:axiom4} states when two variables are individually conditionally independent of each other.

\autoref{ax:axiom4} states the following. Suppose $i$ and $j$ are two variables that the DM believes do not cause each other. Then, the DM should believe that, conditional on $i$'s and $j$'s causes, $i$ and $j$ are independent. Consequently, variables that are neither causes nor consequences of $i$ may only be informative about $i$ if their information is mediated via the causes (and therefore, consequences) of $i$. In this regard, the information contained in $Ca(i)\cup Co(i)$ is a sufficient statistic for $i$: information contained in any other variable is a garbling of the information already present in $Ca(i)\cup Co(i)$.\footnote{While the axiom does not explicitly reference the consequences of $i$, the consequences of $i$ are implicitly represented as they are the causes of variables downstream from $i$. For instance, in \autoref{fig:figaxiom4}, $c$ is a consequence of $i$ and shows up in the axiom as part of the causes of $j$.}

\begin{Axiom}[Causes are sufficient statistics]\label{ax:axiom4}
For all $i,j\in\pazocal{N}$, and for all $ \pazocal{K}\subset\pazocal{N}\setminus\{i,j\}$, if $i\notin Ca(j)$ and $j\notin Ca(i)$, then
\begin{eqnarray*}
i\perp_{\pazocal{K}} j \vert \left(Ca(i)\cup Ca(j)\right)\setminus \pazocal{K}.
\end{eqnarray*}
\end{Axiom}

\indent \autoref{fig:figaxiom4} below illustrates \autoref{ax:axiom4}. Suppose a DM has to predict the value of $i$ and contemplates purchasing information about the realization of $j$. Suppose the DM already knows the realizations of $i$'s and $j$'s causes, $x_b$ and $x_c$. \autoref{ax:axiom4} says the DM should not pay for information about $j$'s realization: because any information $j$ contains about $i$ should be mediated by $\{b,c\}$, knowing the values of $b$ and $c$ means $j$ is uninformative about $i$. Thus, additional information about $j$ should not be valuable.

\begin{figure}[H]
\centering
\scalebox{1}{%
\begin{tikzpicture}[thick]
 nodes
\draw (0,0) node[left=3pt] {$b$};
\draw (0,2) node[left=3pt] {$i$};
\draw (2,0) node[right=3pt] {$j$};
\draw (2,2) node {$c$};
\draw (4,2) node[right=3pt] {$k$};
\draw[->] (0,0.5)--(0,1.5);
\draw[->] (0.5,2)--(1.5,2);
\draw[->] (2,1.5)--(2,0.5);
\draw[->] (0.5,0)--(1.5,0);
\draw[->] (2.5,2)--(3.5,2);
\draw[->] (2.5,0.5)--(3.5,1.5);
\end{tikzpicture}}
\caption{$i$ and $j$ are independent conditional on their respective causes.}\label{fig:figaxiom4}
\end{figure}

 We now present \autoref{ax:axiom0}, which guarantees that the state space does not omit any relevant variables. In line with \cite{savage1972foundations}, we then say that the state space is \emph{complete}:

\begin{Axiom}[Complete state space]\label{ax:axiom0}
$(\forall i,j\in\pazocal{N})$, $(\forall x_{\pazocal{N}\setminus\{i\}}\in X_{\pazocal{N}\setminus\{i\}})$, and $(\forall f,g\in\mathbb{R}^{X_i})$; if $j\in Ca(i)$, then $f\succ_{x_{\pazocal{N}\setminus\{i,j\}}, x_j}g\Leftrightarrow \mathbbm{1}_{x_j}f\succ_{x_{\pazocal{N}\setminus\{i,j\}}}\mathbbm{1}_{x_j}g$.
\end{Axiom}

\indent \autoref{ax:axiom0} states that for any two variables, $i$ and $j$, once the value of all other variables is fixed, causation and conditional independence coincide. To see this, suppose $j$ causes $i$. Then, $j$ contains information about $i$, so the DM should be willing to pay a small enough $\varepsilon$ to learn the value of $j$. But should the DM pay more for \emph{observing $j$'s realization}, $x_j$, or for the ability to \emph{intervene} $j$ at $x_j$? The term $f\succ_{x_{\pazocal{N}\setminus\{i,j\}}, x_j}g$ shows the DM's choice when they intervene $j$ at value $x_j$, whereas the term $\mathbbm{1}_{x_j}f\succ_{x_{\pazocal{N}\setminus\{i,j\}}}\mathbbm{1}_{x_j}g$ is the DM's choice when they condition on the realization $x_j$. Because the choice is the same in both settings, the value of information (i.e., the fee charged to this DM in exchange for information on $j$'s value) is the same whether the value of $j$ is observed or intervened. This equivalence is justified only if the state space includes all relevant variables--i.e., if the state space is complete. Analogously, if neither variable caused the other, $i$ and $j$ would be both conditionally and intervention independent of each other. 
%

\indent Completeness of the state space does \emph{not} imply that the DM must know what all the relevant variables are. For instance, assume that Alex from the introduction is worried that the interaction among ability, education and lifetime earnings might be affected by some other variable. Concretely, Alex thinks that some other variable might influence education: Alex does not know what this variable is, but believes that it exists. For concreteness, denote this variable as an ``unknown but possibly existing variable''. \autoref{ax:axiom0} says that Alex's state space should include such a variable. Therefore, the state space should not be $A\times E\times L$, but rather, $A\times E\times L\times U$, where $U$ stands for ``unknown but possibly existing variable''. In short, \autoref{ax:axiom0} allows the econometrician to add variables that act as proxies for unknown shocks to the system. Indeed, modeling a potential unknown confounder as exogenous noise shocks is a common way to proceed in empirical studies that is accommodated in the state space of our model.

Our final axiom, \autoref{ax:axiom2}, guarantees acyclicity of the causal graph, $G(\overline{\succ})$. The axiom states that the DM's state space only includes logically independent variables.

\begin{Axiom}[Existence of primitives]\label{ax:axiom2}
For all $\pazocal{I}\subset\pazocal{N}$, there exists $i\in\pazocal{I}$ such that $Ca(i)\cap\pazocal{I}=\emptyset$
\end{Axiom}

\indent Under \autoref{ax:axiom2}, if the DM is asked to explain the relation between variables in $\pazocal{I}$, and only those in $\pazocal{I}$, the DM has an explanation that involves at least one exogenous primitive relative to $\pazocal{I}$. Models without primitives describe identities rather than relations among logically independent variables. 

%
\indent When analyzing \autoref{ax:axiom2} we must not confuse a functional relation between variables--such as an equilibrium condition--with intrinsic causal statements. For example, consider a general equilibrium model with aggregate demand curve $D$ and aggregate supply curve $S$. The equilibrium is defined as follows: $(p^*,q^*)$ constitutes an equilibrium if $D(p^*)=q^*$ and $S(p^*)=q^*$. Note that this is a definition; as such, the \emph{equilibrium} price and \emph{equilibrium} quantity are not logically independent. These equations describe the values one should ultimately expect for prices and quantities but are silent regarding the mechanism that generated them. This silence precisely motivates the equilibrium convergence literature. For example, a tat\^onnement convergence process is compatible with the general equilibrium equations without invoking feedback loops: a DM posits that prices in period $t$ cause quantities in period $t$ (via consumer/producer optimization) and that quantities in period $t$ cause prices in period $t+1$ (through a Walrasian auctioneer process that increases/decreases the price in response to excess demand/supply). That the system stabilizes at a point where $p_t=p_{t+1}=p^*$ and $q_t=q_{t+1}=q^*$ does not imply that equilibrium prices cause equilibrium quantities or vice versa. In short, functional equations describe numerical relations between variables, but are not causal statements.

\indent While Axioms \ref{ax:axiom3} through \ref{ax:axiom2} are our basic axioms, \autoref{ax:axiom6} is a supplementary axiom that is relevant for \autoref{th:newtheorem2}. We present it here in the interest of containing all axioms within a single section. We delay a full discussion of \autoref{ax:axiom6} until \autoref{sec:statistics}, where we use it to state \autoref{th:newtheorem2}.

\begin{Axiom}\label{ax:axiom6}
$(\forall i\in\pazocal{N})$, $(\forall\pazocal{J}\subset\pazocal{N}\setminus\{i\})$ $(\forall f,g\in\mathbb{R}^{X_i})$, $(\forall x_{Ca(i)\cup\pazocal{J}}\in X_{Ca(i)\cup\pazocal{J}})$,
\begin{eqnarray*}
\mathbbm{1}_{\{x_{Ca(i)}\}}f\succ\mathbbm{1}_{\{x_{Ca(i)}\}}g\Leftrightarrow \mathbbm{1}_{x_{Ca(i)\setminus\pazocal{J}}}f\succ_{x_{\pazocal{J}}}\mathbbm{1}_{x_{Ca(i)\setminus\pazocal{J}}}g.
\end{eqnarray*}
\end{Axiom}

\section{Results}\label{sec:results}
 In this section we present and discuss our two main results, Theorems \ref{th:theorem1} and \ref{th:newtheorem2}. Our first main result is Theorem \ref{th:theorem1}, stated below.

\begin{Theorem}\label{th:theorem1}
Let $\overline{\succ}$ satisfy \autoref{ass:ass1}. The following are equivalent:
\begin{enumerate}[label=(\roman*)]
\item Axioms \ref{ax:axiom3} through \ref{ax:axiom2} hold, and
\item $G(\overline{\succ})$ is a DAG and represents $\overline{\succ}$ in the sense of \autoref{def:representation2}.
\end{enumerate}
Furthermore, if $G$ also represents $\overline{\succ}$ in the sense of \autoref{def:representation2}, then $G=G(\overline{\succ})$.
\end{Theorem}

 \autoref{th:theorem1} provides a foundation for the analysis of causality that addresses the issues mentioned in the introduction: defining ``causality'', specifying its relation to probabilistic independence, and representing both causality and probabilistic independence in the same model. We now discuss how \autoref{th:theorem1} addresses these challenges.

\paragraph{Defining and representing causality} Any model on causality must begin with a precise definition of causality. Presumably, this definition is connected to the independence properties of the relevant variables, which can be graphically represented in a DAG. A natural intuition is that causality will inherit this graphical representation, but \autoref{fig:reverse} shows this is false. Two different DAGS--i.e., two different sets of arrows--can represent the same set of conditional independence properties; therefore, arrows cannot have intrinsic causal meaning. Consequently, the exact assumptions on how conditional independence and causality interact will determine whether the given definition of causality admits a DAG representation. \autoref{th:theorem1} shows the exact conditions on conditional independence and Bayesian causality under which Bayesian causality admits a DAG representation. 

\indent This provides one of our main contributions to the literature originating with \cite{pearl1995causal}. In these papers, ``causality'' is never defined. Instead, it is assumed that the DM has some definition of ``causality'' in mind, and that this definition of ``causality'' can be represented in a DAG. In our paper, we provide a precise definition of causality--Bayesian causality--and we show the exact axioms under which Bayesian causality can be represented in a DAG.

\indent \autoref{th:theorem1} also implies that Bayesian causality is the only definition of causality with a DAG representation compatible with $\overline{\succ}$. Suppose $C$ is some other definition of causality, and let $G(C)$ be the causal DAG of $C$: that is, $G(C)=(\pazocal{N},\{(i,j): \text{$i$ is a $C$ cause of $j$}\})$. Then, if $G(C)$ represents $\overline{\succ}$ in the sense of \autoref{def:representation2}, then $G(C)=G(\overline{\succ})$; that is, $i$ is a $C$ cause of $j$ if and only if $i$ is a Bayesian cause of $j$.

\paragraph{Why DAGs?} Let us consider alternative models of causality that cannot be represented via a DAG. Such a model should involve a precise definition of ``causality'', $C$. That is, a binary relation, $C$, exists such that statements ``$i$ causes $j$'' are representable as $(i,j)\in C$. There are two main ways in which $C$ might not be representable via a DAG. First, if $C$ coincides with Bayesian causality, \autoref{th:theorem1} implies that the DM's behavior violates our axioms. Our axioms, however, only impose natural assumptions: that the causes and consequences of a variable are a sufficient statistic for the information about that variable, that the state space is complete, and that each (sub)model has at least one primitive. Second, $C$ may not coincide with Bayesian causality. In general, economics differentiates ``causality'' from conditional dependencies because the former is relevant for policy action whereas the latter is not. We believe Bayesian causality is a reasonable formal expression of this idea. As such, \autoref{th:theorem1} implies that DAGs are a natural language in which to express causal models.

\paragraph{The role of each axiom} Axioms \ref{ax:axiom3}--\ref{ax:axiom4} together imply that the causes (and therefore the consequences) of a variable $i$ are a sufficient statistic for that variable. Any other information we can obtain about $i$ is a garbling of the information already contained in $Ca(i)\cup Co(i)$. This delivers an (undirected) graph that represents the conditional independence properties of the DM's beliefs, but the direction of causality is undetermined.\footnote{As mentioned in the introduction, conditional independence may be defined both via directed and undirected graphs.} That is, if an undirected link between $i$ and $j$ exists, then one is a (Bayesian) cause of the other, but Axioms \ref{ax:axiom3}--\ref{ax:axiom4} cannot tell us which is the cause and which the consequence. \autoref{ax:axiom0} determines the direction of causality. That an undirected link between $i$ and $j$ exists means $i$ and $j$ are not independent, even after conditioning on all other variables. Because conditional independence is a symmetric relation we cannot know which is the cause. However, intervention independence is not symmetric. For two variables, $i$ and $j$, consider the $\pazocal{N}\setminus\left\{i,j\right\}$ intervention preferences over $\mathbb{R}^{X_{\{i,j\}}}$. Now consider the intervention preferences on $\mathbb{R}^{X_i}$ after also intervening $j$, and compare those to the Savage conditional beliefs over $\mathbb{R}^{X_i}$. If $j$ was the cause of $i$, these two should coincide, but this is only guaranteed if $j$ is the cause and not the consequence. Thus \autoref{ax:axiom0} permits identification of the direction of causality. Finally, \autoref{ax:axiom2} implies acyclicity of the DM's causal DAG, $G(\overline{\succ})$.

\paragraph{The ``causal effect'' of a variable on another} Finally, our model differentiates between the definition of (Bayesian) causality, which focuses on the direct impact that one variable has on the DM's beliefs about another variable, and common usages of the term ``causal effect'', which may incorporate indirect effects as well. On the one hand, the definition of Bayesian causality follows a ceteris paribus approach: do changes to one variable affect the DM's beliefs about another variable once we exclude any intermediary (or confounding) effects coming from other variables? On the other hand, a DM (say, an applied economist) might want to understand how interventions on one variable affect beliefs about other variables \emph{through all possible channels}. For example, the DM may be a labor economist who believes that education policies have no direct effect on lifetime earnings, but that education policies have an \emph{indirect} effect on lifetime earnings: higher education leads to higher ability levels which, in turn, lead to higher lifetime earnings. That is, while education is not a Bayesian cause of earnings, education still has a ``causal effect'' on earnings.

\indent  Via the DAG formalism, our model makes the distinction between ``cause'' and ``causal effect'' transparent. The ceteris paribus approach is important to get a definition of Bayesian causality. By Theorem \ref{th:theorem1}, Bayesian causality provides the collection of all arrows pointing from Bayesian cause to Bayesian consequence; these are the elemental building blocks of all causal effects between variables. Having constructed these arrows, we can construct all paths joining one variable to another. These paths consequently encode all causal effects, both direct and indirect, through which one variable affects beliefs about another. For the interested reader, Appendix \ref{app:causal-effect} formally defines direct causal effects and indirect causal effects, and illustrates the difference between them via examples.


\subsection{Identification of intervention beliefs} \label{sec:statistics}

 In this section, we address the last issue posed in the introduction: can we express Bayesian causality in terms of probabilistic independence? In principle, a DM needs to form a collection of beliefs, represented by a collection of probabilities, $(\mu_p)_{p\in\pazocal{P}}$, elicitation of which requires observing behavior in the full domain $\left\{(p,f): p\in\pazocal{P},\:,f\in\pazocal{F}(p)\right\}$. However, beliefs elicited in the non-intervention problem--i.e., beliefs elicited from the standard Savage preference, $\succ$--may encode all the information contained in each $(\mu_{p})_{p\in\pazocal{P}}$. We then say intervention beliefs are \emph{identified}. Identification means the DM only really needs to form beliefs as in the traditional Savage model; all other intervention beliefs are derived from the standard setting. 

\indent Below, we present an informal statement of \autoref{th:newtheorem2}: \autoref{ax:axiom6} is necessary and sufficient for intervention beliefs to be identified. Example \ref{ex:identification} illustrates how \autoref{ax:axiom6} allows identification of intervention beliefs. The example illustrates that \autoref{ax:axiom6} connects the structure of arrows in $G(\overline{\succ})$ to two different operations we may carry out on intervention beliefs, denoted \emph{exchange} and \emph{eliminate}. We conclude the section with a formal statement of \autoref{th:newtheorem2}: \autoref{ax:axiom6} is necessary and sufficient for the exchange and eliminate rules to hold. 

\begin{Theorem*}[Informal statement of \autoref{th:newtheorem2}]
Let $\overline{\succ}$ satisfy \autoref{ass:ass1}, and Axioms \ref{ax:axiom3} through \ref{ax:axiom2}. The following are equivalent:
\begin{enumerate}[label=(\roman*)]
\item \autoref{ax:axiom6} holds, and
\item Intervention beliefs are identified: for each $p\in\pazocal{P}$, $\mu_p$ is written purely in terms of $\mu$.
\end{enumerate}
\end{Theorem*}

 To illustrate the role of \autoref{ax:axiom6} in identifying intervention beliefs, we first remind the reader of the statement of \autoref{ax:axiom6}. We follow up with a discussion on what \autoref{ax:axiom6} implies in terms of the DM's behavior, and conclude with a discussion of how \autoref{ax:axiom6} permits identification of intervention beliefs.

\begin{axiom6restated}
$(\forall i\in\pazocal{N})$, $(\forall\pazocal{J}\subset\pazocal{N}\setminus\{i\})$, $(\forall f,g\in\mathbb{R}^{X_i})$, $(\forall x_{Ca(i)\cup\pazocal{J}}\in X_{Ca(i)\cup\pazocal{J}})$,
\begin{eqnarray} \label{ax6eq}
\mathbbm{1}_{\{x_{Ca(i)}\}}f\succ\mathbbm{1}_{\{x_{Ca(i)}\}}g\Leftrightarrow \mathbbm{1}_{x_{Ca(i)\setminus\pazocal{J}}}f\succ_{x_{\pazocal{J}}}\mathbbm{1}_{x_{Ca(i)\setminus\pazocal{J}}}g.
\end{eqnarray}
\end{axiom6restated}

\indent  \autoref{ax:axiom6} states the following two decision problems are equivalent. Given a variable $i$ and acts $f,g\in\mathbb{R}^{X_i}$, the first problem is to choose $f$ or $g$ when their payments are contingent on the causes of $i$ obtaining a particular value, $x_{Ca(i)}$. In the second decision problem, the DM intervenes on a subset of causes of $i$ (say, intervening $\pazocal{J}\subset Ca(i)$ to the value $x_\pazocal{J}$), and the payments of $f$ and $g$ are now contingent on the values of the non-intervened causes, $x_{Ca(i)\setminus\pazocal{J}}$, being realized. From a numerical standpoint, both these situations result in the same value for the causes of $i$ (namely, $x_{Ca(i)}$); the difference is \emph{how} those values are obtained. In the first problem, it is simply by selecting a standard Savage conditional act; in the second problem, it is by a combination of interventions and Savage conditional acts. Because \autoref{ax:axiom6} requires that these two problems be treated identically, \autoref{ax:axiom6} implies that the only aspect of interventions that matters is the value the intervention sets for the variable. In other words, the act of intervening on a variable does not, in itself, change the DM's structural view of the world.

\indent We use \autoref{fig:axiom6fig} below to illustrate \autoref{ax:axiom6}:

\begin{figure}[H]
\centering
\scalebox{1}{%
\begin{tikzpicture}[thick]
\draw (-0,0) node {$j$};
\draw (0,4) node {$k$};
\draw (4,4) node {$w$};
\draw (4,0) node {$i$};
\draw[->] (0,3.5)--(0,.5);
\draw[->] (0.5,4)--(3.5,4);
\draw[->] (0.5,0)--(3.5,0);
\draw[->] (4,3.5)--(4,.5);
\end{tikzpicture}}
\caption{Observing or intervening on $j$ makes the DM update differently about $k$. This difference in updating may affect the DM's beliefs about $i$. }\label{fig:axiom6fig}
\end{figure}

\indent First, we explain why the left hand side of \autoref{ax6eq} involves all causes of $i$. Suppose that a DM has to choose between two acts over $i$ (say, $f,g\in\mathbb{R}^{X_i}$), the payments of which are contingent on $j$ taking value $x_j$. That is, the DM has to choose between $\mathbbm{1}_{x_j}f$ and $\mathbbm{1}_{x_j}g$. Note that $\{j\}$ is a strict subset of $Ca(i)$. Observing that $j$ takes the value $x_j$ gives the DM information about the value of $k$; in turn, this information about $k$ gives the DM information about $w$, which ultimately gives the DM information about $i$. Thus, observing that $j$ took the value $x_j$ is informative about $i$ in two ways: directly, because $j\in Ca(i)$, and indirectly, via $k$ and $w$. If the DM intervenes on $j$ at value $x_j$, the DM receives the same direct information about $i$ but loses the indirect information mediated via $k$ and $w$. Thus, the DM could say that $\mathbbm{1}_{x_j}f\succ\mathbbm{1}_{x_j}g$ but $g\succ_{x_j}f$. Clearly, observing $x_j$ or intervening on variable $j$ and moving it to value $x_j$ are different problems in terms of the DM's updating. Hence, \autoref{ax:axiom6} is reasonable only when we fix \emph{all} causes of $i$.

\indent Now, consider the situation above but where the payments of $f$ and $g$ involve \emph{all} causes $i$, $j$ \emph{and} $w$. On the left hand side of \autoref{ax6eq}, the DM observes the realizations $x_j$ and $x_w$, and chooses between $\mathbbm{1}_{x_j,x_w}f$ and $\mathbbm{1}_{x_j,x_w}g$. On the right hand side of \autoref{ax6eq}, the DM intervened $j$ at level $x_j$, and then chooses between $\mathbbm{1}_{x_w}f$ and $\mathbbm{1}_{x_w}g$. Following the logic in the previous paragraph, intervening $j$ means that the DM loses the information $j$ provides about $i$ via $w$. However, because the realization of $w$ is observed in both decision problems, no actual information is lost. Thus, the DM has the same information in both problems so they should make the same decision in both problems. This result is precisely what \autoref{ax:axiom6} requires.

\indent Both of the above discussions addressed $\pazocal{J}\subset Ca(i)$, but to complete our discussion of \autoref{ax:axiom6}, we must allow that $\pazocal{J}$ contains non-causes of $i$. \autoref{ax:axiom6} states that once we know the value of all the causes of $i$, intervening variables that are not causes of $i$ are uninformative about $i$. In \autoref{fig:axiom6fig}, if an act's payments are contingent on $x_w$ and $x_j$, then intervening and shifting the value of $k$ to some $x_k$ is uninformative about $i$.

 \autoref{ax:axiom6} implies the following conditions on the DM's beliefs: $(\forall i\in\pazocal{N})$, $(\forall\pazocal{J}\subset\pazocal{N}\setminus\{i\})$, $(\forall x_i\in {X_i})$, $(\forall x_{Ca(i)\cup\pazocal{J}}\in X_{Ca(i)\cup\pazocal{J}})$,
\begin{eqnarray}
\mu(x_i\vert x_{Ca(i)}) = \mu_{x_\pazocal{J}}(x_i\vert x_{Ca(i)\setminus\pazocal{J}}). \label{ax6eq2}
\end{eqnarray}

In words, this means that we may freely exchange conditioning on the causes of a variable with conditioning on a subset of causes and intervening the remaining causes. On the left hand side of \autoref{ax6eq2}, we have the DM's beliefs on $i$ conditional on the causes of $i$. On the right hand side of equation \ref{ax6eq2}, we are intervening a subset of variables, $\pazocal{J}\subset \pazocal{N}$, and we are conditioning on the remaining causes, $Ca(i)\setminus\pazocal{J}$. \autoref{ax:axiom6} implies these two operations coincide. From  \autoref{ax6eq2}, we obtain the exchange and eliminate rules, as illustrated in \autoref{ex:identification} below: 

\begin{Example}\label{ex:identification}
Consider the following DAG: education causes higher ability levels, which cause higher lifetime earnings. However, socio-economic status of childhood neighborhood ($S$) is a common cause of both education and lifetime earnings. The DM wants to calculate the effects that an education policy has on lifetime earnings; i.e., how $\mu_{x_E}(x_L)$ varies with $x_E$ for each $x_L$. Under \autoref{ax:axiom6}, we obtain the following formula:
\begin{eqnarray}
\mu_{x_E}(x_L)=\sum_{x_S}\mu(x_L\vert x_E, x_S)\mu(x_S). \label{formulaid}
\end{eqnarray} 

\indent This particular formula is relevant as it allows identification of $\mu_{x_E}$ without using the distribution of $A$. If the DM forms their Bayesian beliefs by updating some prior using the data available to them, this means the DM does not need data on (the possibly unobservable) $A$ to identify the effect of an education policy on lifetime earnings. Hence, this is an identification formula that is independent from the DM's bayesian prior over $A$. 
\begin{figure}[H]
\centering
\begin{tikzpicture}[thick]
\draw (3,3) node {$S$};
\draw (1,1) node {$E$};
\draw (3,1) node {$A$};
\draw (5,1) node {$L$};
\draw (5,0.3) node {};
\draw (5,3.3) node {};

\draw[black, ->] (2.5,2.5) -> (1.5,1.5);
\draw[black, ->] (1.5,1) -> (2.5,1);
\draw[black, ->] (3.5,1) -> (4.5,1);
\draw[black, ->] (3.5,2.5) -> (4.5,1.5);
\end{tikzpicture}
\caption{An extended confounding triangle.}\label{fig:compare}
\end{figure}
\end{Example}

\indent We now use \autoref{ax:axiom6} to derive \autoref{formulaid} step by step. First, we write $\mu_{x_E}(x_L)=\sum_{x_S}\mu_{x_E}(x_L\vert x_S)\mu_{x_E}(x_S)$. Therefore, it suffices to identify the expressions $\mu_{x_E}(x_L\vert x_S)$ and $\mu_{x_E}(x_S)$.
\paragraph{Identifying $\mu_{x_E}(x_S)$ in Example \ref{ex:identification}} Setting $i=S$,  $Ca(S)=\emptyset$, and $\pazocal{J}=\{E\}$, \autoref{ax:axiom6} implies $\mu_{x_E}(x_S)=\mu(x_S)$. Notice that the expressions $\mu_{x_E}(x_S)$ and $\mu(x_S)$ are the same, except that the intervention $x_E$ is eliminated. 
\paragraph{Identifying $\mu_{x_E}(x_L\vert x_S)$ in \autoref{ex:identification}} We now show that $\mu_{x_E}(x_L\vert x_S)=\mu(x_L\vert x_E, x_S)$. We can derive this by applying \autoref{ax:axiom6} and \autoref{ax:axiom4}. If we select $i=L$, $Ca(i)=\{A, S\}$, and $\pazocal{J}=\{E\}$, \autoref{ax:axiom6} implies $\mu_{x_E}(x_L\vert x_S, x_A )=\mu(x_L\vert x_S, x_A )$. Furthermore, \autoref{ax:axiom4} implies that $\mu_{x_E}(x_A\vert x_S)=\mu_{x_E}(x_A)$; i.e. $A$ and $S$ are independent after intervening $E$. We can apply \autoref{ax:axiom6} to $\mu_{x_E}(x_A)$ and obtain $\mu_{x_E}(x_A)=\mu(x_A\vert x_E)=\mu(x_A\vert x_S, x_E)$, where the last equality follows because $S$ and $A$ are independent conditional on $E$. Together, we obtain the following identification:
\begin{eqnarray*}
\mu_{x_E}(x_L\vert x_S)&=&\sum_{x_A}\mu_{x_E}(x_L\vert x_S, x_A)\mu_{x_E}(x_A\vert x_S)\\
&=&\sum_{x_A}\mu(x_L\vert x_S, x_A)\mu(x_A\vert x_S, x_E)\\
&=&\sum_{x_A}\mu(x_L\vert x_S, x_E, x_A)\mu(x_A\vert x_S, x_E)\\
&=&\mu(x_L\vert x_E, x_S),
\end{eqnarray*}
where the second line follows from applying our axioms, and the third line holds because $E$ and $L$ and independent given $A$ and $S$. Comparing $\mu_{x_E}(x_L\vert x_S)$, and $\mu(x_L\vert x_E, x_S)$ the difference is that we exchanged the intervention belief $\mu_{x_E}$ with the conditional probability $\mu(\cdot\vert x_E, \cdot)$. 
\paragraph{Completing the identification in \autoref{ex:identification}} Because $\mu_{x_E}(x_L)=\sum_{x_S}\mu_{x_E}(x_L\vert x_S)\mu_{x_E}(x_S)$, $\mu_{x_E}(x_S)=\mu(x_S)$ and $\mu_{x_E}(x_L\vert x_S)=\mu(x_L\vert x_E, x_S)$, we obtain $\mu_{x_E}(x_L)=\sum_{x_S}\mu(x_L\vert x_E, x_S)\mu(x_S)$ as desired.

As illustrated above, \autoref{ax:axiom6} identifies intervention beliefs by carrying out two operations: interchanging interventions with conditioning, and eliminating interventions altogether. The first operation holds when we can write $\mu_{x_\pazocal{J}}(x_\pazocal{H}\vert x_\pazocal{K})=\mu(x_\pazocal{H}\vert x_{\pazocal{J}},\:x_\pazocal{K})$. On the left hand side of the equation, $\pazocal{J}$ is a set of intervened variables, whereas on the right hand side $\pazocal{J}$ is not intervened, but we are conditioning on it. The second operation holds when $\mu_{x_\pazocal{J}}(x_\pazocal{H}\vert x_\pazocal{K})=\mu(x_\pazocal{H}\vert x_\pazocal{K})$. In the left hand side of the equation $\pazocal{J}$ is an intervened set of variables, whereas in the right hand side we dropped the intervention altogether. By eliminating the intervention terms, both of these operations allow us to identify intervention beliefs.
\paragraph{Exchanging interventions and conditional probability} The exchange rule provides a graphical test for when interventions and conditional probability can be exchanged; that is, $\mu_{x_\pazocal{J}}(x_\pazocal{H}\vert x_\pazocal{K})=\mu(x_\pazocal{H}\vert x_{\pazocal{J}},\:x_\pazocal{K})$.

\begin{Rule}[Exchange]\label{rule:rule1}
Let $ \{i_1\}, I_0, J, K\subset\pazocal{N}$ be four disjoint sets of variables. Let $G(\overline{\succ})$ be the DM's causal DAG and construct $G_{I_0^{in}, i_1^{out}}(\overline{\succ})$ from $G(\overline{\succ})$ by eliminating all arrows coming into $I_0$ and all arrows coming out of $i_1$. 

If, in $G_{I_0^{in}, i_1^{out}}(\overline{\succ})$, $J$ is independent of $i_1$ conditional on $K\cup I_0$, then $\mu_{x_{I_0}, x_{i_1}}(x_j\vert x_K) = \mu_{x_{I_0}}(x_J\vert x_K, x_{i_1})$.
\end{Rule}

 In \autoref{ex:identification}, eliminating the arrow $E\rightarrow A$ makes $E$ and $L$ independent conditional on $S$, so Rule \ref{rule:rule1} delivers $\mu_{x_E}(x_L\vert x_S)=\mu(x_L\vert x_E, x_S)$ as we illustrated above. 

\paragraph{Eliminating interventions} The second rule of causal calculus provides a graphical test for when interventions can be eliminated; that is, $\mu_{x_\pazocal{J}}(\cdot)=\mu(\cdot)$.

\begin{Rule}[Eliminate]\label{rule:rule2}
Let $ \{i_1\}, I_0, J, K\subset\pazocal{N}$ be four disjoint sets of variables. Let $G(\overline{\succ})$ be the DM's causal DAG. Construct $G_{I_0^{in}, i_1^{in}}(\overline{\succ})$ from $G$ by eliminating all arrows into $I_0$ and $i_1$, and construct $G_{I_0^{in}}(\overline{\succ})$ from $G$ by eliminating all arrows into $I_0$.
\begin{enumerate}
\item Suppose $i_1$ is not an indirect cause of any variable in $K$. If $J$ is independent of $i_1$ conditional on $K\cup I_0$ according to $G_{I_0^{in}, i_1^{in}}(\overline{\succ})$ then $\mu_{x_{I_0}, x_{i_1}}(x_j\vert x_K) = \mu_{x_{I_0}}(x_J\vert x_K)$.
\item  Suppose $i_1$ is an indirect cause of some variable in $K$. If $J$ is independent of $i_1$ conditional on $K\cup I_0$ according to $G_{I_0^{in}}(\overline{\succ})$ then $\mu_{x_{I_0}, x_{i_1}}(x_j\vert x_K) = \mu_{x_{I_0}}(x_J\vert x_K)$.
\end{enumerate} 
\end{Rule}

  In \autoref{ex:identification}, $E$ is not an indirect cause of $S$. We can therefore eliminate the arrow $S\rightarrow E$. Once we do this, $E$ and $S$ are independent, so Rule \ref{rule:rule2} delivers $\mu_{x_E}(x_S)=\mu(x_S)$, as illustrated above. 

\autoref{th:newtheorem2} connects the Rules \ref{rule:rule1} and \ref{rule:rule2} to \autoref{ax:axiom6}. As mentioned before, this shows us the added necessary and sufficient conditions required for a causal DAG to be useful in the identification of intervention probabilities.\footnote{The reader familiar with Pearl's work may recall the existence of a third rule of causal calculus. However, it can be shown that Rules \ref{rule:rule1} and \ref{rule:rule2} imply the remaining one, and so we omit it from the discussion.}

\begin{Theorem}\label{th:newtheorem2}
Let $\overline{\succ}$ satisfy \autoref{ass:ass1}, and Axioms \ref{ax:axiom3} through \ref{ax:axiom2}. The following are equivalent:
\begin{enumerate}[label=(\roman*)]
\item \autoref{ax:axiom6} holds, and
\item Rules \ref{rule:rule1} and \ref{rule:rule2} hold.
\end{enumerate}
\end{Theorem}

\paragraph{Recovering $\overline{\succ}$ from $\succ$} To recover the DM's preferences, $\overline{\succ}$, when we only observe their non-intervention preferences, $\succ$, we need additional information: the direction of causality and whether or not Rules \ref{rule:rule1} and \ref{rule:rule2} hold. \autoref{th:newtheorem2} complements \autoref{th:theorem1} in recovering $\overline{\succ}$ from $\succ$. \autoref{th:theorem1} says we need to complement our observations of $\succ$ with an assumption on the representing causal DAG, $G$. \autoref{th:newtheorem2}  says we need to complement the information on $\succ$ and $G$ with the assumption that Rules \ref{rule:rule1} and \ref{rule:rule2} hold. We elaborate below.

\indent First, not observing $\overline{\succ}$ means we lose information on the direction of causality. From $\succ$ in the Savage framework we obtain the DM's probabilistic beliefs, $\mu$, representable by some collection of DAGs. These DAGs only encode information on the conditional independence structure of $\mu$. However, conditional independence is a symmetric relation--$i$ is independent of $j$ if and only if $j$ is independent of $i$--while Bayesian causality is not. \autoref{th:theorem1} shows that Axioms \ref{ax:axiom3} through \ref{ax:axiom2} are necessary and sufficient to recover the direction of causality. Consequently, if we complement our observation of $\succ$ with a specific choice for a representing DAG, we partially recover $(\succ_p)_{p\in\pazocal{P}}$.

\indent Second, not observing $\overline{\succ}$ means we lose information on the intervention beliefs, $(\mu_p)_{p\in\pazocal{P}}$. Suppose we complement our observation of $\succ$ with a choice of causal DAG, $G$. We then know that the DM's preference, $(\succ_p)_{p\in\pazocal{P}}$, is compatible with Axioms \ref{ax:axiom3} through \ref{ax:axiom2}. However, \autoref{th:newtheorem2} implies this is not enough to recover the DM's intervention beliefs, $(\mu_p)_{p\in\pazocal{P}}$. Suppose then that we are willing to assume that Rules \ref{rule:rule1} and \ref{rule:rule2} hold. Then we can express the DMs intervention beliefs, $(\mu_p)_{p\in\pazocal{P}}$, in terms of $\mu$. Because $\mu$ is information we have, because of $\succ$, we can now recover the DM's full preference structure, $\overline{\succ}$.

\paragraph{\autoref{th:newtheorem2} and causal inference} Relative to the literature initiated by Pearl's work, we provide necessary and sufficient conditions for intervention beliefs--consequently, causality--to be identified. Indeed, as seen in  \autoref{th:theorem1}, causal DAGs are well-defined representations of causal models, but \autoref{th:newtheorem2} implies that intervention beliefs need not be identified. Only under \autoref{ax:axiom6} can we identify intervention beliefs, and they are identified via Rules \ref{rule:rule1} and \ref{rule:rule2}. 

\indent In the causal inference literature, Pearl proves analogous version of Rules \ref{rule:rule1} and \ref{rule:rule2} using a formalism he called \emph{do-probability}. A contribution of \autoref{th:newtheorem2} is that we derive Rules \ref{rule:rule1} and \ref{rule:rule2} directly from an axiom on behavior, rather than the do-probability formalism. In Appendix \ref{app:conections} we show that do-probabilities and intervention beliefs coincide if, and only if, Axiom \ref{ax:axiom6} holds. Consequently, Pearl's rules of causal calculus and Rules \ref{rule:rule1} and \ref{rule:rule2} coincide and they identify Bayesian causality if, and only if, Axiom \ref{ax:axiom6} holds.

\paragraph{Rules \ref{rule:rule1} and \ref{rule:rule2} and truncations of the DM's causal DAG} Finally, it is natural to wonder why Rules \ref{rule:rule1} and \ref{rule:rule2} require truncating the DM's causal DAG, $G(\overline{\succ})$, and why those specific truncations are the truncations that yield the result in \autoref{th:newtheorem2}. Appendix \ref{app:truncations} provides an intuitive discussion on why these specific truncations of the DM's causal DAG are relevant when carrying out causal identification. 

\section{Conclusions}

\indent In this paper we complement \citeauthor{savage1972foundations}'s choice over acts framework to incorporate choice over both \emph{policy interventions} and acts defined on the non-intervened variables. In this way, we can compare a DM's conditional preferences with their intervention preferences, thereby distinguishing a subjective belief that two variables are correlated from a subjective belief that one variable causes another. 

\indent Equipped with this extended choice domain, the paper makes three contributions. First, we provide an explicit definition of what it means for a DM to behave as if one variable causes another--which we denote \emph{Bayesian causality}. This definition formally captures a prevalent idea in applied economic literature: the difference between correlation and causation is that the latter is relevant for policy recommendations, whereas the former is not. Second, we provide an axiomatic characterization of the DM's causal model. The representation is a DAG representation, which connects our model to a wider literature on causal inference. Furthermore, in our decision theoretic framework, we show that  representing a causal model via a DAG is equivalent to defining causality as Bayesian causality. Lastly, we show additional axioms under which intervention beliefs--therefore, causal effects--are identified from probabilistic (in)dependence. The identification results, together with the DAG representation of the DM's causal model, provide a decision theoretic foundation for models such as \cite{pearl1995causal}. Furthermore, our characterization separates two important components of graphical methods for causal inference: the representation of causality via a DAG, and the identification methodology used to elicit causality from probabilistic (in)dependence. 

\indent From an applied perspective, our paper provides tools for carrying out causal inference. Indeed, the DM in our model may be thought of as an econometrician carrying out an empirical study, and the econometrician's subjective beliefs are those they arrive at by analyzing the available data. To the extent the econometrician finds our axioms normatively compelling, our paper suggests that \citeauthor{pearl1995causal}'s model of causal inference is a faithful representation of an econometrician's subjective causal model.

%
%

\bibliographystyle{ecta}
\bibliography{causality}

\appendix
\section{Omitted Proofs}\label{App:profs}
\indent In this appendix we provide the proofs for Theorems \ref{th:theorem1} and \ref{th:newtheorem2}. Both these proofs rely on the notion of d-separation, which in turn depends on whether specific paths in a DAG has colliders. Therefore, we begin by introducing the definitions of colliders and d-separation, which we then use to prove the theorems.

\subsection{General graph-theoretic definitions}

\begin{Definition}[Undirected paths]\label{def:colliders}
Let $G=(V,E)$ be a DAG. An \emph{undirected path in $G$} is a tuple $p=(v_1,...,v_N)$ where $v_n\in V$ for each $n$ and either $(v_n,v_{n+1})\in E$--that is, $v_n\rightarrow v_{n+1}$--or   $(v_{n+1},v_{n})\in E$--that is, $v_n\leftarrow v_{n+1}$. 
\end{Definition}

\begin{Definition}[Colliders and tail-to-tail nodes]\label{def:colliders}
Let $G=(V,E)$ be a DAG, and $p$ be an undirected path in $G$. We say that $v^*$ is a \emph{collider} in $p$ if for some $n$ we have $(v_{n-1},v_n,v_{n+1})$ where $v_n=v^*$, $v_{n-1}\rightarrow v^* \leftarrow v_{n+1}$. We say $v^*$ is a \emph{tail-to-tail} node if $v_{n-1}\leftarrow v^* \rightarrow v_{n+1}$.
\end{Definition}

\begin{Definition}[d-separation]\label{def:dseparation}
Let $\pazocal{I},\pazocal{J},\pazocal{K}\subset \pazocal{N}$ be three disjoint sets of variables. We say that $\pazocal{K}$ d-
separates $\pazocal{I}$ from $\pazocal{J}$ if for each undirected path between a variable in $\pazocal{I}$ and a variable in $\pazocal{J}$, one of the following properties holds:
\begin{enumerate}
\item There is a node $w$ along the path such that $w$ is a collider such that $w\notin \pazocal{K}$ and $\pazocal{K}\subset ND(w)$.
\item There is a node $w$ along the path such that $w$ is not a collider and such that $w\in \pazocal{K}$.
\end{enumerate}
If either of the above properties hold, we say $w$ \emph{blocks} $p$, of that $p$ is \emph{blocked by} $w$.
\end{Definition}

It is known from the conditional independence literature that if $\pazocal{K}$ d-separates $\pazocal{I}$ from $\pazocal{J}$  then $\pazocal{I}$ is independent of $\pazocal{J}$ conditional on $\pazocal{K}$ (see, e.g., \citealp{dawid1979conditional,pearl1995causal,lauritzen1990independence}). We use this fact to prove Theorem \ref{th:theorem1}.

\subsection{Proof of Theorem \ref{th:theorem1}}
\begin{lemma}\label{lemma:lemmadseparation1}
Fix $\pazocal{K}\subset\pazocal{N}$ and $x_\pazocal{K}\in X_{\pazocal{K}}$. Let $G_{\pazocal{K}}$ represent $\succ_{x_\pazocal{K}}$. For each $i\in\pazocal{N}$, $Ca(i)\setminus\pazocal{K}$ d-separates $\{i\}$ from $ND(i)\setminus\pazocal{K}\equiv\{\hat{j}\in\pazocal{K}^\complement:\text{ $i$ is not an indirect cause of $\hat{j}$}\}$.
\end{lemma}
\begin{proof}
Choose $j\in\{\hat{j}\in\pazocal{K}^\complement:\text{ $i$ is not an indirect cause of $\hat{j}$}\}$.
Choose an undirected path $p$ from $j$ to $i$.
That is, $p=(i_0,...,i_N)$ where $i_0=j$, $i_N=i$, and, for each $n\in\{1,...,N\}$, either $(i_{n-1},i_n)\in E$ or $(i_n,i_{n-1})\in E$.
First, since $i$ is not an indirect cause of $j$, then $p$ cannot be a directed path from $i$ to $j$. That is, $p$ cannot be such that $(i_{n},i_{n-1})\in E$ for each $n$.
Second, if $p$ is a directed path from $j$ to $i$ (that is, $(i_{n-1},i_n)\in E$ for each $n$), then $p$ is blocked by $i_{N-1}\in Ca(i)\setminus
\pazocal{K}$.
Third, assume that $p$ is not directed in any direction.
Then, $p$ has colliders and/or tail-to-tail nodes.
Let $i_n$ be the last node that is either a collider or a tail-to-tail node.
Let $q=(i_n,...,i_N)$ be the path starting at $i_n$.
By the definition of $i_n$, $q$ must be directed.
Assume that $q$ is directed from $i_n$ to $i$.
Then, $i_n$ is tail-to-tail.
Then, $p$ is blocked by $i_{N-1}$.
Finally, assume that $q$ is directed from $i$ to $i_n$.
Then, $i_n$ is a collider.
If $i_n\in Ca(i)\setminus\pazocal{K}$, then $(i_n,i, q)$ is a cycle.
Thus, $i_n\notin Ca(i)\setminus\pazocal{K}$.
By a similar argument, no descendant of $i_n$ can be in $Ca(i)\setminus\pazocal{K}$.
Therefore, $i_n$ blocks $p$.
Since each path joining $j$ to $i$ is blocked, this concludes the proof.
\end{proof}

\begin{lemma}\label{lemma:ax1}
Assume $G(\overline{\succ})$ is a DAG that represents $\overline{\succ}$. Then \autoref{ax:axiom3} holds.
\end{lemma}
\begin{proof}
This follows from the minimality requirement in \autoref{def:representation0}. 
Pick $i\in\pazocal{N}$.
Pick any $\pazocal{H}$, $\pazocal{J}$ as in \autoref{ax:axiom3}.
Because $G(\overline{\succ})$ represents $\succ$, then $\mu_{x_{\pazocal{H}}}(x_{-\pazocal{H}})=\Pi_{j\notin\pazocal{H}} \mu_{x_{-\pazocal{H}}}(x_j\vert x_{Ca(j)})$, and this representation is minimal.
That is, if $B\subsetneq Ca(i)$, then $\mu_{x_{\pazocal{H}}}(x_{-\pazocal{H}})\neq\Pi_{j\neq i} \mu_{x_{-\pazocal{H}}}(x_j\vert x_{Ca(j)})\mu_{x_{\pazocal{H}}}(x_i\vert B)$. 
Suppose, contrary to \autoref{ax:axiom3}, that $i\perp_\pazocal{H}(Ca(i)\setminus(\pazocal{J}\cup\pazocal{H}))\vert\pazocal{J}$.
Then, there is a value $x_{\pazocal{H}}\in X_{\pazocal{H}}$  such that
$\mu_{x_{\pazocal{H}}}(x_i\vert x_{Ca(i)\setminus\pazocal{H}})=\mu_{x_{\pazocal{H}}}(x_i\vert x_{\pazocal{J}})$, which contradicts the minimality requirement.
\end{proof}

\begin{lemma}\label{lemma:axsecret}
Assume $G(\overline{\succ})$ is a DAG that represents $\overline{\succ}$. Then \autoref{ax:secret} holds.
\end{lemma}
\begin{proof}
Suppose that $i$, $\pazocal{J}$, and $\pazocal{K}$ are as in \autoref{ax:secret}.
Because $j\in ND(i)$ for each $j\in\pazocal{J}\setminus\pazocal{K}$, \autoref{lemma:lemmadseparation1} implies that $i\perp_{\pazocal{K}}\pazocal{J}\setminus\pazocal{K}\vert Ca(i)\setminus\pazocal{K}$.
Alternatively, the factorization formula $\mu_{x_{\pazocal{K}}}(x_{-\pazocal{K}})=\Pi\mu_{x_{\pazocal{K}}}(x_i\vert Ca(i)\setminus\pazocal{K})$ implies $i$ is independent of its non-descendants given $i$'s causes. Then, it is independent of any subset $\pazocal{J}$ of its non-descendants. 
\end{proof}

\begin{lemma}\label{lemma:ax2}
Assume $G(\overline{\succ})$ is a DAG that represents $\overline{\succ}$. Then \autoref{ax:axiom4} holds.
\end{lemma}
\begin{proof}
Take $i$ and $j$ such that $i\notin Ca(j)$ and $j\notin Ca(i)$.
Let $\pazocal{K}\subset\pazocal{N}\setminus \{i,j\}$.
We prove the lemma by showing that $Ca(i)\cup Ca(j)\setminus\pazocal{K}$ blocks all paths from $i$ to $j$ in $G_{\pazocal{K}}$.
Let $p$ is an undirected path from $i$ to $j$ that is not blocked, and enumerate $p=(i,i_0,..., i_T,j)$.
Because $p$ is not blocked, $i_0\notin Ca(i)$ and $i_T\notin Ca(j)$.
Indeed, if either $i_0\in Ca(i)$ or $i_T\in Ca(j)$, then either $i_0$ is not a collider or $i_T$ is not a collider, thus implying that $p$ is blocked by $Ca(i)\cup Ca(j)$.
Therefore, $p$ has a collider, because $i\rightarrow i_0$ and $i_T\leftarrow j$.
Let $n$ be the smallest number such that $i_n$ is a collider and $m$ be the largest number such that $i_m$ is a collider (possibly $n=m$).
Then, the path $(i,i_0,...,i_n)$ is a directed path from $i$ to $i_n$ so $i_n\notin Ca(i)$ (analogously, $i_m\notin Ca(j)$ because the path $(j,i_T,..,i_m)$ is directed) by acyclicity of $G(\overline{\succ})$. 
Because of this, and because $p$ is not blocked, the following must be true: (i) $i_n\in Ca(j)$, and (ii) $i_m\in Ca(i)$. Then, the directed path that goes from $i$ to $i_n$, jumps to $j$, returns to $i_m$, and skips back to $i$ is a cycle.\footnote{Formally, this is the path $q=(i, i_0,... i_n, j, i_T, i_{T-1}, ..., i_m, i)$}
This constitutes a contradiction.
Thus, every path $p$ from $i$ to $j$ is blocked by $Ca(i)\cup Ca(j)$.
Thus, \autoref{ax:axiom4} holds.
\end{proof}

\begin{lemma}\label{lemma:ax3}
Assume $G(\overline{\succ})$ is a DAG that represents $\overline{\succ}$. Then \autoref{ax:axiom0} holds.
\end{lemma}
\begin{proof}
\autoref{ax:axiom0} is a direct consequence of item (ii) in \autoref{def:representation2}.
Indeed, let $i,j\in\pazocal{N}$ and let $f\in\mathbb{R}^{X_i}$.
Then, $f\succ_{x_{-\{i,j\}}} g\Leftrightarrow \sum u(f(x_i))\mu_{x_{-\{i,j\}}}(x_i)> \sum u(g(x_i))\mu_{x_{-\{i,j\}}}(x_i)$.
Because of item (ii) in \autoref{def:representation2}, we get  $\mu_{x_{-\{i,j\}}}(x_i)=\mu_{x_{-\{i\}}}(x_i\vert x_j)$.
Therefore, $\sum u(f(x_i))\mu_{x_{-\{i,j\}}}(x_i)> \sum u(g(x_i))\mu_{x_{-\{i,j\}}}(x_i)\Leftrightarrow \sum u(f(x_i))\mu_{x_{-\{i\}}}(x_i\vert x_j)> \sum u(g(x_i))\mu_{x_{-\{i,j\}}}(x_i\vert x_j)$.
Because $\sum u(f(x_i))\mu_{x_{-\{i\}}}(x_i\vert x_j)> \sum u(g(x_i))\mu_{x_{-\{i,j\}}}(x_i\vert x_j) \Leftrightarrow  \mathbb{1}_{x_j}f\succ_{x_{-\{i,j\}}} \mathbb{1}_{x_j}g$ we obtain $f\succ_{x_{-\{i,j\}}} g\Leftrightarrow \mathbb{1}_{x_j}f\succ_{x_{-\{i,j\}}} \mathbb{1}_{x_j}g$, so \autoref{ax:axiom0} is satisfied.
\end{proof}

\begin{lemma}\label{lemma:ax4}
Assume $G(\overline{\succ})$ is a DAG that represents $\overline{\succ}$. Then \autoref{ax:axiom2} holds.
\end{lemma}
\begin{proof}
By contradiction, suppose that \autoref{ax:axiom2} did not hold.
Then, there would exist $i$ and a sequence $(i,i_1,...,i_T,i)$ such that $i\in Ca(i_1)$, for all $t\in\{1,...,T-1\}$, $i_t\in Ca(i_{t+1})$, and $i_T\in Ca(i)$.
Thus, $((i,i_1), ..., (i_{t-1},i_t), ..., (i_T,i))$ is a cycle in $G(\overline{\succ})$, which is a contradiction because $G(\overline{\succ})$ is acyclic.
\end{proof}

For the next lemma we use the following notation. Given a directed path, $p=(p_1,...,p_{L+1} )$, between two variables, $p_1$ and $p_{L+1}$, let $\#p$ be the length of path $p$. That is,  $\#p\equiv L$. Given a variable $t\in \pazocal{N}$ let $O(t)=\max\{\#p: \text{$p$ is a directed path from $v$ to $t$ and $Ca(v)=\emptyset$}\}$. That is, $O(t)$ is the longest length of a path from a primitive, $v$, to the variable $t$. By convention, if $Ca(t)=\emptyset$ we let $O(t)=0$. We refer to $O(t)$ as the \emph{order of} $t$.

\begin{lemma}\label{lemma:markov}
Assume that Axioms \ref{ax:secret} and \ref{ax:axiom4} hold. Then, $(\forall \pazocal{K}\subset\pazocal{N})$, $(\forall i\notin\pazocal{K}$, $(\forall j\notin\pazocal{K}\cup Ca(i)\text{ such that }i\notin ICa(j))$ ,
\begin{eqnarray*}
i\perp_{\pazocal{K}} j \vert (Ca(i)\setminus\pazocal{K}).
\end{eqnarray*}
\end{lemma}
\begin{proof}
Fix $\pazocal{K}\subset\pazocal{N}$ and choose an arbitrary $i\notin\pazocal{K}$.
We proceed by complete induction on the order of $j$: for all $M$, if $j\notin\pazocal{K}\cup Ca(i)$, $O(j)=M$, and $i\notin ICa(j)$, then $i\perp_{\pazocal{K}} j \vert (Ca(i)\setminus\pazocal{K})$.
First, consider any $j\notin\pazocal{K}\cup Ca(i)$ with $M=0$. 
\autoref{ax:axiom4} implies $i\perp_{\pazocal{K}} j\vert (Ca(i)\cup Ca(j))\setminus\pazocal{K}$.
Because $O(j)=0$ then $Ca(j)=\emptyset$ and so $i\perp_{\pazocal{K}} j\vert Ca(i)\setminus\pazocal{K}$ as desired.
Now, assume the inductive hypothesis is true for all $m\leq M$, and let $j\notin\pazocal{K}\cup Ca(i)$ be of order $M+1$.
Then, \autoref{ax:axiom4} implies $i\perp_{\pazocal{K}}j\vert (Ca(i)\cup Ca(j))\setminus \pazocal{K} $.
Furthermore, for each $w\in Ca(j)\setminus (Ca(i)\cup\pazocal{K})$, $O(w)\leq M$.\footnote{Note that $i\notin ICa(w)$. Because $w\in Ca(j)$, if $i$ was an indirect cause of $w$ then $i$ would be an indirect cause of $j$, a contradiction. }
Thus, by the inductive hypothesis, $i\perp_{\pazocal{K}}w\vert Ca(i)\setminus\pazocal{K}$.
By \autoref{ax:secret}, because $i\perp_{\pazocal{K}}w\vert Ca(i)\setminus\pazocal{K}$ for each $w\in Ca(j)\setminus (Ca(i)\cup\pazocal{K})$ then $i \perp_{\pazocal{K}} (Ca(j)\setminus (Ca(i)\cup\pazocal{K})) \vert Ca(i)\setminus\pazocal{K}$.
Because we also have  $i\perp_{\pazocal{K}}j\vert (Ca(i)\cup Ca(j))\setminus\pazocal{K}$, then $i\perp_{\pazocal{K}} j \vert (Ca(i)\setminus\pazocal{K})$.
By induction, for all $j\notin\pazocal{K}\cup Ca(i)$, if $i\notin ICa(j)$ then $i\perp_{\pazocal{K}} j \vert (Ca(i)\setminus\pazocal{K})$.
\end{proof}

\begin{Corollary}\label{coro:markov}
As a corollary of \autoref{lemma:markov}, we obtain the Markov property in \cite{lauritzen1990independence}, $i\perp_{\pazocal{K}} ND(i)\setminus Ca(i)\vert Ca(i)\setminus\pazocal{K}$.
Indeed, \autoref{lemma:markov} delivers $i\perp_{\pazocal{K}}j\vert Ca(i)\setminus\pazocal{K}$ for each $j\in ND(i)\setminus Ca(i)$. Applying \autoref{ax:secret} gives $i\perp_{\pazocal{K}} (ND(i)\setminus Ca(i))\vert Ca(i)\setminus\pazocal{K}$.
\end{Corollary}

\begin{lemma}\label{lemma:theorem1}
Suppose $\overline{\succ}$ satisfies \autoref{ass:ass1}. If $\;\;\overline{\succ}$ satisfies \autoref{ax:axiom3} through \ref{ax:axiom2}, $G(\overline\succ)$ represents $\overline{\succ}$.
\end{lemma}
\begin{proof}
Construct $G(\overline{\succ})$ by setting $Pa(i)=Ca(i)$.
By \autoref{ax:axiom2}, $G(\overline{\succ})$ is acyclic.
Indeed, if for some length $k\in\mathbb{N}$, there were a cycle $e=((i_1,i_2), (i_2,i_3), ..., (i_k,i_1))$, then $i_1$ would be an indirect cause of itself.
Choose any set $\pazocal{K}\subset\pazocal{N}$ and any realization $x_{\pazocal{K}}\in X_{\pazocal{K}}$.
We need to show that $\mu_{x_{\pazocal{K}}}(x_{-\pazocal{K}})=\Pi_{i\notin\pazocal{K}}\mu_{x_{\pazocal{K}}}(x_i\vert Ca(i)\setminus\pazocal{K})$.
By \autoref{lemma:markov} and Corollary \ref{coro:markov}, $\mu_{x_{\pazocal{K}}}$ satisfies the directed local Markov property in \cite{lauritzen1990independence}.
Therefore, $\mu_{x_{\pazocal{K}}}(x_{-\pazocal{K}})=\Pi_{i\notin\pazocal{K}}\mu_{x_{\pazocal{K}}}(x_i\vert Ca(i)\setminus\pazocal{K})$.\footnote{For completeness, we provide a proof sketch of  \cite{lauritzen1990independence} 's argument. Without loss of generality, label variables in $\pazocal{N}\setminus\pazocal{K}$ so that $j\leq i$ implies $i$ is not an indirect cause of $j$; this can be done by ranking variables according to their order in the graph, $O(\cdot)$. By our enumeration, $\{j\notin\pazocal{K}:j<i \}\subset\{j\in\pazocal{N}:ij\in ND(i)\}$. By the chain rule, we get $\mu_{x_{\pazocal{K}}}(x_{-\pazocal{K}})=\Pi_{i\notin\pazocal{K}}\mu_{x_{\pazocal{K}}}(x_i\vert ND(i)\setminus\pazocal{K})$ Because $i$ is independent of its non-descendants given $Ca(i)$ then we can write $\mu_{x_{\pazocal{K}}}(x_i\vert ND(i)\setminus\pazocal{K})=\mu_{x_{\pazocal{K}}}(x_i\vert Ca(i)\setminus\pazocal{K})$ and this completes the argument.}
We now check the minimality of $Ca(i)$.
Assume that $\pazocal{J}\subsetneq Ca(i)$.
\autoref{ax:axiom3} states that $i\not\perp_\pazocal{K} \pazocal{J}\vert Ca(i)\setminus(\pazocal{K}\cup\pazocal{J}))$.
Thus, the factorization formula $\mu_{x_{\pazocal{K}}}(x_{-\pazocal{K}})=\Pi_{i\notin\pazocal{K}}\mu_{x_{\pazocal{K}}}(x_i\vert Ca(i)\setminus\pazocal{K})$ is minimal.
\end{proof}
\begin{lemma}\label{lemma:unique}
Let $\overline{\succ}$ be the DM's preferences, and let $G(\overline{\succ})=(\pazocal{N},E)$ be the directed graph defined by setting $Pa(i)=Ca(i)$ for each $i\in\pazocal{I}$. If $\overline{\succ}$ satisfies \autoref{ass:ass1}, and \autoref{ax:axiom0}, then the following are true:
\begin{enumerate}
\item If $G=(\pazocal{N},F)$ is a directed graph that represents $\overline{\succ}$, then $(j,i)\in F \Rightarrow j\in Ca(i)$.
\item If $G=(\pazocal{N},F)$ is a directed graph that represents $\overline{\succ}$, then $j\in Ca(i)\Rightarrow (j,i)\in F$ or $i\in Ca(j)$.
\end{enumerate}
\end{lemma}
\begin{proof}
Let $\overline{\succ}$ be as in the statement of the lemma, $G(\overline{\succ})$ be the directed graph defined by setting $Pa(i)=Ca(i)$ for each $i\in\pazocal{N}$, and $G=(\pazocal{N},F)$ be any other directed graph that represents $\overline{\succ}$.
For each $\pazocal{I}\subset \pazocal{N}$ and each realization $x_{\pazocal{I}}\in X_{\pazocal{I}}$, let $\mu_{x_{\pazocal{I}}}\in \Delta(X_{-\pazocal{I}}) $ represent beliefs obtained from $\succ_{x_{\pazocal{I}}}$.
Such a representation exists from \autoref{ass:ass1}.
\newline We first show that $j\in Ca(i) \Rightarrow (j,i)\in F$ or $i\in Ca(j)$.
If $j\in Ca(i)$, then the function $T:X_j\rightarrow \mathbb{R}$ defined as $T(x_j)=\mu_{x_{-\{i,j\}}, x_{j}}(x_i)$ is not constant in $x_j$.
Additionally, by \autoref{ax:axiom0}, $\mu_{x_{-\{i,j\}}}(x_i\vert x_j)=T(x_j)$.
Thus, $i$ and $j$ are not independent after intervening on $\{i,j\}^\complement$.
Because $G$ represents $\overline{\succ}$, then $G_{-\{i,j\}}$ represents $\succ_{-\{i,j\}}$.
Thus, either $(i,j)\in F$ or $(j,i)\in F$ (if not, $G_{-\{i,j\}}$ would treat $i$ and $j$ as independent, which is a contradiction).
If $(j,i)\in F$, the proof concludes.
Therefore, let $(j,i)\notin F$ so that $(i,j)\in F$.
Because $G$ represents $\overline{\succ}$, this means that $\mu_{x_{-\{i,j\}}}(x_j\vert x_i)=\mu_{x_{-\{j\}}}(x_j)$.
Because $\mu_{x_{-\{i,j\}}}(x_j\vert x_i)$ is not constant in $i$, this equation implies $i\in Ca(j)$, as desired.
\newline We now show $(j,i)\in F\Rightarrow j\in Ca(i)$.
First, note that for all $x\in X$, $\mu_{x_{-\{i,j\}}}(x_i, x_j)=\mu_{x_{-\{i,j\}}}(x_j)\mu_{x_{-\{i,j\}}}(x_i\vert x_j)$.
Because $G$ represents $\overline{\succ}$, $(j,i)\in F$ and the minimality condition in \autoref{def:representation0}, jointly imply that $i$ and $j$ are not independent after intervening on $\{i,j\}^\complement$.
That is, $\mu_{x_{-\{i,j\}}}(x_i\vert x_j)$ is not constant in $x_j$.
Moreover, because $G$ represents $\overline{\succ}$ and $(j,i)\in F$, we obtain that $\mu_{x_{-\{i,j\}}}(x_i\vert x_j)=\mu_{x_{-\{i\}}}(x_i)$.
Therefore, there is a value of $x_{-\{j\}}$ for which $T(x_j)=\mu_{x_{-\{i\}}}(x_i)$ is not constant in $x_j$.
Therefore, $j\in Ca(i)$.
\end{proof}
\begin{Remark}\label{remark:unique}
We only used \autoref{ax:axiom0} to prove the second statement in \autoref{lemma:unique}. Thus, without \autoref{ax:axiom0}, any representing graph, $F$, must include the causal links in the sense of \autoref{def:causality} (\emph{i.e.}, $(j,i)\in F\Rightarrow j\in Ca(i)$), but $F$ could omit some arrows. However, only arrows involved in 2-cycles are omitted. If $G$ is acyclic then $G=G(\overline{\succ})$.
\end{Remark}

\begin{proof}[Proof of \autoref{th:theorem1}]
That the representation implies the axioms is a direct consequence of Lemmas \ref{lemma:ax1} through \ref{lemma:ax4}.
By \autoref{lemma:theorem1} the axioms imply the representation.
The uniqueness claim follows directly from \autoref{lemma:unique} and Remark \ref{remark:unique}.
\end{proof}

\subsection{Proof of \autoref{th:newtheorem2}}\label{app:theorem2}
\begin{lemma}\label{lemma:rules}
Assume $G(\overline{\succ})$ represents $\overline{\succ}$. If Rules \ref{rule:rule1} and \ref{rule:rule2} hold then \autoref{ax:axiom6} holds.
\end{lemma}
\begin{proof}
Fix any $\pazocal{K}$. Let $i\in\pazocal{N}$, $\pazocal{J}\subset\pazocal{N}_{-i}$.
To show \autoref{ax:axiom6} hold it suffices to show $\mu(x_i\vert Ca(i))=\mu_{x_{\pazocal{J}}}(x_i\vert x_{Ca(i)\setminus\pazocal{J}})$.
Let $\pazocal{J}^*=\pazocal{J}\cap Ca(i)$.
Rule \ref{rule:rule2} implies the following:
\begin{eqnarray}
\mu_{\pazocal{J}^*\cup(\pazocal{J}\setminus Ca(i))}(x_i\vert x_{Ca(i)\setminus \pazocal{J} })=\mu_{\pazocal{J}^*}(x_i\vert x_{Ca(i)\setminus \pazocal{J}}).\label{bla}
\end{eqnarray}
We need to show that for each $j\in\pazocal{J}\setminus Ca(i)$, the set $Ca(i)=(Ca(i)\setminus\pazocal{J})\cup \pazocal{J}^*$ blocks all paths from $i$ to $j$ in the corresponding truncation.
Take any $j\in\pazocal{J}\setminus Ca(i)$. 
First, suppose $j$ is an indirect cause of some $v\in Ca(i)\setminus \pazocal{J}$ then $j$ is an indirect cause of $i$.
Therefore, $j$ is a non-descendant of $i$.
As such, $Ca(i)$ block all paths from $j$ to $i$ and so the second item of Rule \ref{rule:rule2} delivers \autoref{bla}.
Suppose now that $j$ is not an indirect cause of any $v\in Ca(i)\setminus \pazocal{J}^*$.
Then, $j$ is a non-descendant of $i$ in $G_{\{j\}^{in}}(\overline{\succ})$ because there are no arrows into $j$.
Thus, $Ca(i)$ blocks all paths from $i$ to $j$ in $G_{(Ca(i)\setminus\pazocal{J}^*)^{in},\{j\}^{in}}(\overline{\succ})$ and so the first item of Rule \ref{rule:rule2} delivers \autoref{bla}.

Rule \ref{rule:rule1} implies the following:
\begin{eqnarray}
\mu_{\pazocal{J}^*}(x_i\vert x_{Ca(i)\setminus \pazocal{J}})=\mu(x_i\vert x_{(Ca(i)\setminus\pazocal{J})\cup\pazocal{J}^* }), \label{blano}\\
\mu_{\pazocal{J}^*}(x_i\vert x_{Ca(i)\setminus \pazocal{J}})=\mu(x_i\vert x_{Ca(i)}). \label{bla2}
\end{eqnarray}

Indeed, $Ca(i)\setminus \pazocal{J}^*$ d-separates $i$  from $\pazocal{J}^*$ in $G_{(\pazocal{J}^*)^{out}}$ and so the independence clause for Rule \ref{rule:rule1} holds.

Putting together Equations \ref{bla} and \ref{bla2} we obtain, $\mu(x_i\vert x_{Ca(i)})=\mu_{x_{\pazocal{J}}}(x_i\vert x_{Ca(i)\setminus\pazocal{J}})$ as desired.
\end{proof}
The next lemma is important in its own right for two reasons. First, it gives a formal meaning to the informal description of \autoref{th:newtheorem2}: \autoref{do-probability} below shows that intervention beliefs are expressible purely in terms of non-intervention beliefs. In particular, if we know the DM's Savage preferences, $\succ$, and we assume \autoref{ax:axiom6} holds (alternatively, Rules \ref{rule:rule1} and \ref{rule:rule2} hold), we can recover their intervention beliefs. Second, it states that our intervention probabilities coincide with Pearl's do-probabilities (see Appendix \ref{app:conections} for a more detailed discussion of intervention beliefs and do-probabilities).
\begin{lemma}\label{lemma:do-probability}
Suppose \autoref{ax:axiom3} through \ref{ax:axiom6} hold. Then, for all $x\in X$, and all $\pazocal{I}\subset \pazocal{N}$,
\begin{eqnarray}
\mu_{x_{\pazocal{I}}}(x_{\pazocal{N}\setminus \pazocal{I}})=\frac{\mu(x)}{\Pi_{i\in\pazocal{I}} \mu(x_i\vert x_{Ca(i)})}. \label{do-probability}
\end{eqnarray}
\end{lemma}
\begin{proof}
Because Axioms \ref{ax:axiom3} through \ref{ax:axiom2} hold, then $\mu$ can be factorized through the causal graph, $G(\overline{\succ})$.
That is, we get $\mu(x)=\Pi \mu(x_n\vert x_{Ca(n)})$.
Therefore, $\frac{\mu(x)}{\Pi_{i\in\pazocal{I}} \mu(x_i\vert x_{Ca(i)})}=\Pi_{n\notin\pazocal{I}} \mu(x_n\vert x_{Ca(n)})$.
Using \autoref{ax:axiom6} we obtain $\Pi_{n\notin\pazocal{I}} \mu(x_n\vert x_{Ca(n)})=\Pi_{n\notin\pazocal{I}} \mu_{x_i}(x_n\vert x_{Ca(n)\setminus\{i\}})$.
Then, $\frac{\mu(x)}{\Pi_{i\in\pazocal{I}} \mu(x_i\vert x_{Ca(i)})}= \Pi_{n\notin\pazocal{I}} \mu_{x_i}(x_n\vert x_{Ca(n)\setminus\{i\}})$.
Because $G_{\pazocal{I}^{in}}(\overline{\succ})$ represents $\succ_{x_{\pazocal{I}}}$ then  the factorization formula delivers $\Pi_{n\notin\pazocal{I}} \mu_{x_i}(x_n\vert x_{Ca(n)\setminus\{i\}})=\mu_{x_{\pazocal{I}}}(x_{\pazocal{N}\setminus\pazocal{I}})$.
Hence,  $\frac{\mu(x)}{\Pi_{i\in\pazocal{I}} \mu(x_i\vert x_{Ca(i)})}=\mu_{x_{\pazocal{I}}}(x_{\pazocal{N}\setminus\pazocal{I}})$ as desired.
\end{proof}

\begin{proof}[Proof of \autoref{th:newtheorem2}]
That \autoref{ax:axiom6} implies the rules is a direct consequence of \autoref{lemma:do-probability} and \cite{pearl1995causal}.
\autoref{lemma:do-probability} implies that for all $I_0\subset\pazocal{N}$, $\mu_{I_0}(\cdot)$ coincides with \cite{pearl1995causal}'s do probability $Pr(x_{-I_0}\vert set(X_{I_0}=x_{i_0}))$.
Then, the rules follow from Theorem 3 in \cite{pearl1995causal}.
That the rules imply \autoref{ax:axiom6} follows from \autoref{lemma:rules}.
\end{proof}

\section{More on \autoref{th:newtheorem2} and connections to the Pearl model.}

\indent This section complements the discussion of \autoref{th:newtheorem2} in the main text. First, Appendix \ref{app:truncations} expands on some intuitions presented in the main text. Rules \ref{rule:rule1} and \ref{rule:rule2} apply to a truncated DAG, rather than the original causal DAG elicited from the DM's preferences. In Appendix \ref{app:truncations} we provide an intuition for why these truncations are part of the aforementioned rules. Second, Appendix \ref{app:conections} provides the connections between \autoref{th:newtheorem2} and \cite{pearl1995causal}. In that paper, Pearl defines an object--denoted a \emph{do-probability}--which can be identified from conditional probabilities via the use of two rules --which Pearl calls the \emph{rules of causal calculus}. Appendix \ref{app:conections} shows that our \autoref{ax:axiom6} is equivalent to the condition that do-probabilities coincide with our intervention-beliefs, so that Pearl's rules of causal calculus coincide with Rules \ref{rule:rule1} and \ref{rule:rule2}. Consequently, \autoref{ax:axiom6} shows the exact conditions under which Bayesian causality and Pearl's do-probability analysis coincide. Finally, Appendix \ref{app:causal-effect} discusses applications of \autoref{th:newtheorem2} for inferring the ``causal effect'' of one variable on another. In particular, we argue that the term ``casual effect'' may be understood in different ways, and that the DAG formalism allows us to distinguish different interpretations of the term ``casual effect'' without risking any ambiguity.

\subsection{Truncating DAGs and Rules \ref{rule:rule1} and \ref{rule:rule2}}\label{app:truncations}
In this section we study why truncating DAG $G(\overline{\succ})$ is necessary when applying rules \ref{rule:rule1} and \ref{rule:rule2}.

\paragraph{Comparing information transmission: interventions vs. conditional probability} Consider a DM's causal DAG, depicted in \autoref{fig:compare-b}. Because of Axioms \ref{ax:axiom3} through \ref{ax:axiom4}, the set of causes of a variable is an exhaustive source of information about that variable. Consequently, any information a variable $i$ can generate about a variable $j$ will be transmitted via the causes of $i$ (in \autoref{fig:compare-b}, $a$) , the consequences of $i$ (in \autoref{fig:compare-b}, $n$), or both. 

\begin{figure}[H]
\centering
\begin{tikzpicture}[thick]
\draw (3,3) node {$a$};
\draw (1,1) node {$i$};
\draw (3,1) node {$n$};
\draw (5,1) node {$j$};
\draw (5,0.3) node {};
\draw (5,3.3) node {};

\draw[black, ->] (2.5,2.5) -> (1.5,1.5);
\draw[black, ->] (1.5,1) -> (2.5,1);
\draw[black, ->] (3.5,1) -> (4.5,1);
\draw[black, ->] (3.5,2.5) -> (4.5,1.5);
\end{tikzpicture}
\caption{$i$ provides information about $j$ through $a$ and $n$}\label{fig:compare-b}
\end{figure}

 A decision maker who wants to predict the value of $j$ should pay a small $\varepsilon>0$ to observe the realized value of $i$. First, observing that $i$ took value $x_i$ is informative about what value $n$ might take, because $i$ causes $n$. In turn, this information on $n$ is informative about the value of $j$ because $n$ causes $j$. Second, observing that $i$ took value $x_i$ provides information about the values of $a$ that could have generated $x_i$; in turn, this is informative about the value of $j$. Consequently, the conditional probability $\mu(x_j\vert x_i)$ encodes the information $i$ transmits about $j$ via both the causes of $i$ and the consequences of $i$.

\indent A decision maker who wants to predict the value of $j$ will lose some information if they instead intervene the variable $i$. If $i$ is intervened to level $x_i$, $i$ will still provide the same information about $n$ as it did above, which will ultimately provide the same information about $j$. However, because $i$ is no longer caused by $a$, $i$ no longer provides information about $a$, so the path $i\leftarrow a\rightarrow j$ is no longer informative. In terms of the DAG in \autoref{fig:compare-b} intervening $i$ is \emph{as if} we severed the link $a\rightarrow i$, and so we lose one path through which $i$ is connected to $j$. If the path $i\rightarrow n \rightarrow j$ did not exist, interventions of $i$ would have no value for prediction $j$ because all information $i$ transmits about $j$ is contained in the causes of $i$.

\paragraph{Interventions, conditional probability, and truncating DAGs} From the previous analysis we conclude that interventions of $i$ and conditioning on the value of $i$ are interchangeable when information is transmitted through the consequences of $i$, but not when information is transmitted through the causes of $i$. Furthermore, if $i$ transmits information about $j$ only via the causes of $i$, intervening $i$ is uninformative about $j$. This underpins the intuition behind Rules \ref{rule:rule1} and \ref{rule:rule2}.

\indent In Rule \ref{rule:rule1}, $G(\overline{\succ})_{i^{out}}$ eliminates all arrows emerging from $i$. Because we eliminate all arrows emerging from $i$, this truncation eliminates all information $i$ provides about $j$ via $i$'s consequences, while still allowing $i$ to provide information about $j$ via $i$'s causes. If $i$ and $j$ are independent of each other according to $G(\overline{\succ})_{i^{out}}$ this means that eliminating the consequences of $i$ eliminates all information $i$ contains about $j$. In other words, $i$ is informative about $j$ exclusively through the consequences of $i$. Therefore, interventions on $i$ and conditioning on $i$ are interchangeable: $\mu_{x_i}(x_j)=\mu(x_j\vert x_i)$ holds. If $i$ and $j$ are not independent of each other according to $G(\overline{\succ})_{i^{out}}$, then $i$ is informative about $j$ through $i$'s causes, so $\mu_{x_i}(x_j)=\mu(x_j\vert x_i)$ cannot hold. In the example in \autoref{fig:compare-b}, $a$ is a common cause of $i$ and $j$ in $G(\overline{\succ})_{i^{out}}$. Therefore, $i$ and $j$ are not independent so $\mu_{x_i}(x_j)\neq \mu(x_j\vert x_i)$. Applying Rule \ref{rule:rule1} when some additional variables are intervened eliminates arrows into those variables to indicate that they are now exogenous primitives in the DM's model.

\indent In Rule \ref{rule:rule2} we have two cases. When $i_1$ is an indirect cause of some conditioning variable, $k\in K$, we can derive Rule \ref{rule:rule2} from Rule \ref{rule:rule1}. First we establish that interventions on $i_1$ can be exchanged with conditioning on $i_1$, then we use independence to drop the conditioning on $i_1$. Because the second step is applied after exchanging the intervention on $i_1$ with a conditional on $i_1$, no truncation on $i_1$ is needed. For example, in \autoref{fig:compare-b}, intervening $i$ only generates information on $j$ through the outgoing path $i\rightarrow n \rightarrow j$. So Rule \ref{rule:rule1} applies and we obtain $\mu_{x_i}(x_j\vert x_n)=\mu(x_j\vert x_n, x_i)$. Because $n$ already encodes all the information $i$ contains on $j$, $\mu(x_j\vert x_n, x_i)= \mu(x_j\vert x_n)$. Thus we get Rule \ref{rule:rule2}.\footnote{More generally, if $i_1$ and $j$ are independent conditional on $K\cup I_0$ according to $G(\overline{\succ})_{I_0^{in}}$, then they are also independent according to $G(\overline{\succ})_{ I_0^{in}, i_1^{out} }$, which is the condition for Rule \ref{rule:rule1} to apply. Furthermore, because $i_1$ is independent of $j$ conditional on $K\cup I_0$, we get $\mu_{x_{I_0}}(x_j\vert x_{K}, x_{i_1})=\mu_{x_{I_0}}(x_j\vert x_{K})$. Therefore, $\mu_{x_{I_0}, x_{i_1}}(x_j\vert x_{K})=\mu_{x_{I_0}}(x_j\vert x_{K})$ as desired.} When $i_1$ is not an indirect cause of some $k\in K$, we apply the truncation logic as we did before: intervening $i_1$ is like treating $i_1$ as an exogenous primitive so we eliminate all information $i_1$ provides to $j$ via $i_1$'s causes. If in this truncation, $i_1$ is independent of $j$ (conditional on $K$ after intervening $I_0$), then $i_1$ only transmitted this information to $j$ via its causes, not its consequences. Thus, we can treat the intervention of $i$ as uninformative about $j$; i.e., we can drop the intervention of $i$ when eliciting beliefs about $j$.

\subsection{Intervention beliefs, Markov representations, and do-probabilities: more connections to DAGs in causal inference}\label{app:conections}

 In this section we study the connection between \autoref{th:newtheorem2} and the identification results in \cite{pearl1995causal} and the subsequent literature. In \cite{pearl1995causal}, the DAG formalism is complemented with the assumption that probability distributions admit a \emph{Markov representation}, which we define below. From this Markov representation, \cite{pearl1995causal} defines \emph{do-probailities} (see \autoref{def:dop}). Theorem 3 in \cite{pearl1995causal} shows that Rules \ref{rule:rule1} and \ref{rule:rule2}, which Pearl denotes ``rules of causal calculus'', guarantee that do-probabilities are identified from the primitive probability distributions. In this section, we show that under \autoref{ax:axiom6} do-probabilities are equivalent to intervention beliefs in our paper. As we discuss at the end of this appendix, the equivalence between do-probabilities and intervention beliefs implies that the do-probability is completely characterized by the rules of causal calculus. 

\paragraph{Markov representation} Suppose $P\in \Delta(X_1\times...\times X_N)$ is a probability distribution, and suppose $G$ is a DAG that factorizes $P$. That is, $P$ satisfies that $P(x_1,...,x_N)=\Pi_{n=1}^N P(x_i\vert Pa(i))$ where $Pa(i)$ are the parents of variable $i$ in DAG $G$. Now, suppose that for each $i\in\{1,...,N\}$ there exists a random variable, $\varepsilon_i$, with range in some set $E_i$ and distribution $\phi_i\in \Delta(E_i)$, and a deterministic function, $f:X_{Pa(i)}\times E_i\rightarrow X_i$, that satisfy the following:
\begin{eqnarray}
\phi (\varepsilon_1,...,\varepsilon_N) &=& \Pi_{i=1}^N \phi_i(\varepsilon_i), \label{markov2} \\
P(x_1,...,x_N) &=& \phi(\{\varepsilon_1,...,\varepsilon_N: \: (\forall \: i) \: f_i(X_{Pa(i)}, \varepsilon_i)=x_i\}). \label{markov1}
\end{eqnarray}

\autoref{markov2} states that the ``noise'' terms, $(\varepsilon_i)_{i=1}^N$ are jointly distributed by $\phi$, with marginals $(\phi_i)_{i=1}^N$, and that they are all independent of each other. \autoref{markov1} states that each variable, $x_i$, is a deterministic function of its parents and the random term $\varepsilon_i$. \autoref{markov1} is what motivates the definition of a do-probability, to which we turn next.

\paragraph{Do-probabilities} In \cite{pearl1995causal}, the main objects of interest are \emph{do-probabilities}. Below, we define do-probabilities, illustrate them via a numerical example, and then discuss their conceptual interpretation. 
\begin{Definition}[Do-probability]\label{def:dop} 
Suppose $K$ is a subset of variables. The probability of $x_{-K}$ do-$x_K$, denoted $P(x_{-K}\vert do(X_K=x_K))$ is calculated as follows: first, eliminate from the Markov representation of $P$ all the functions $(f_k)_{k\in K}$; then replace each instance of $X_k$ in the remaining equations with the ``do-values'' $x_K$ and calculate the probability of the remaining variables according to $\phi_{-K}$.
\end{Definition}
\autoref{ex:dop} shows concretely an example of a do-probability:
 \begin{Example}\label{ex:dop}
 Consider the following DAG:
 \begin{figure}[H]
\centering
\begin{tikzpicture}[thick]
\draw (3,3) node {$a$};
\draw (1,1) node {$i$};
\draw (3,1) node {$n$};
\draw (5,1) node {$j$};
\draw (5,0.3) node {};
\draw (5,3.3) node {};

\draw[black, ->] (2.5,2.5) -> (1.5,1.5);
\draw[black, ->] (1.5,1) -> (2.5,1);
\draw[black, ->] (3.5,1) -> (4.5,1);
\draw[black, ->] (3.5,2.5) -> (4.5,1.5);
\end{tikzpicture}
\caption{$i$ provides information about $j$ through $a$ and $n$}\label{fig:compare}
\end{figure}

Suppose $P$ can be factorized by this DAG and, for simplicity, assume each variable takes values $\{0,1\}$. The following is a Markov representation for $P$:
\begin{eqnarray*}
f_a(\varepsilon_a) &=& \varepsilon_a, \\
f_i(a, \varepsilon_i) &=& a*\varepsilon_i,\\
f_j(\varepsilon_j) &=& (1-n)*a*\varepsilon_j,\\
f_n(\varepsilon_n) &= & i* \varepsilon_n, 
\end{eqnarray*}
where the noise terms are distributed as follows
\begin{eqnarray*}
\phi_a(\varepsilon_a=1)=P(a=1),\\
\quad \phi_i(\varepsilon_i=1)=P(i=1 \vert a=1),\\
\phi_i(\varepsilon_j=1)=P(j=1 \vert a=1, n=0),\\
\phi_n(\varepsilon_n=1)=P(n=1 \vert i=1).
\end{eqnarray*}
Suppose we want to calculate $P(x_j=1\vert do(X_i=0))$. We then eliminate the equation $f_i(a, \varepsilon_i) = a*\varepsilon_i$ from the Markov representation and replace $i=0$ everywhere $i$ appears in the remaining expressions. Thus, our modified Markov representation is as follows:
\begin{eqnarray*}
f_a(\varepsilon_a) &=& \varepsilon_a  \quad \phi_a(\varepsilon_a=1)=P(a=1)\\
f_j(\varepsilon_j) &=& (1-n)*a*\varepsilon_j \\
f_n(\varepsilon_n) &= & 0* \varepsilon_n.
\end{eqnarray*}
Using this reduced Markov representation we conclude that $n=0$ with probability $1$, which means $Pr(x_j=1\vert do(X_i=0))=\phi_a(\varepsilon_a=1)\phi_j(\varepsilon_j=1)=P(a=1)P(j=1\vert a=1, n=0)$. 
\end{Example}
 The conceptual idea is that the ``do'' operator removes the variables from the statistical model and replaces them with deterministic values. In this regard, the $f$ functions in a Markov representation are taken to have causal interpretations, with all confounding uncertainty emanating from the $\varepsilon$ terms. Under this interpretation, $i$ is a cause of $j$ if and only if $i\in Pa(j)$ and $Pr(x_j\vert do(x_{Pa(j)} ))$ is a non-trivial function of $x_i$ for some fixed value of $x_{Pa(j)\setminus\{i\}}$. Of course, in our model the DAG $G$ used to define a Markov representation is not given but, rather, an object we elicit from preferences. Furthermore, nothing guarantees that the DM's beliefs admit a Markov representation nor that intervention beliefs coincide with do-probabilities.
 
 \paragraph{Do-probabilities and intervention beliefs} We now connect do-probabilities with intervention beliefs. As discussed above, the interpretation of a do-probability is that, if the functional equations in a Markov representation are interpreted causally, then do-probabilities are a natural quantification of causality. However, \autoref{th:theorem1} shows that intervention beliefs are a natural representation of the DM's causal model, but nothing guarantees that the DM's intervention beliefs have a Markov representation. Thus, there is a gap between Bayesian causality and the do-probability quantification.
 
 \indent The results in Appendix \ref{app:theorem2} show that \autoref{ax:axiom6} (together with the other Axioms) implies and is implied by the condition that intervention beliefs are do-probabilities. This gives us the exact condition that bridges the gap between Bayesian causality and do-probability calculus. In particular, because \autoref{ax:axiom6} is also equivalent to the rules of causal calculus, then the rules of causal calculus are an alternative characterization of do-probabilities. 
 
 \indent We show the connection between intervention beliefs and do-probabilities in two ways: first, we construct a Markov representation explicitly. Second, we show that \autoref{ax:axiom6} is equivalent to a condition in \cite{pearl1995causal} that characterizes do-probability.
 
\paragraph{Constructing a Markov representation for $\overline{\succ}$} Assume that \autoref{ax:axiom6} holds; we use this to construct a Markov representation of $\mu_{x_{\pazocal{J}}}$ for each $\pazocal{J}\subset\pazocal{N}$. By Theorem \ref{th:theorem1}, Axioms \ref{ax:axiom3} through \ref{ax:axiom2} imply that $G(\overline{\succ})$ represents $\overline{\succ}$.
For each $i\in\pazocal{N}$, let $\varepsilon_i\sim U[0,1]$.
For each realization $x_i\in X_i$ and each $x_{Pa(i)}\in X_{Pa(i)}$, let $I(x_i,x_{Pa(i)})\subset [0,1]$ be an interval of length $\mu_{x_{Pa(i)}}(x_i)$.
Because $\sum_{x_i\in X_i}\mu_{x_{Pa(i)}}(x_i)=1$ for each $x_{Pa(i)}$, then $I(\cdot,x_{Pa(i)})$ can be chosen to form a partition of $[0,1]$.
Fix any variable $i\in\pazocal{N}$, and let $h_i(x_{Pa(i)},\varepsilon_i)=\sum_{x_i\in X_i}x_i\mathbbm{1}_{I(x_i,x_{Pa(i)})}(\varepsilon_i)$.
By construction, $(G,(h_1,...,h_N),(\varepsilon_1,...,\varepsilon_N))$ is a Markov representation of the beliefs elicited from $\overline{\succ}$.
Choose any $\pazocal{J}\subset\pazocal{N}$ and any $i\notin\pazocal{J}$.
By \autoref{ax:axiom6}, for each $x_{i}\in X_i$ and each $x_{Ca(i)\cup\pazocal{J}}\in X_{Ca(i)\cup\pazocal{J}}$, we obtain
\begin{eqnarray}
\mu_{x_\pazocal{J}}(x_i\vert x_{Ca(i)\setminus\pazocal{J}})=\mu(x_i\vert x_{Ca(i)}). \label{eq:th2eq1}
\end{eqnarray}
Our Markov representation implies
\begin{eqnarray}
\mu(x_i\vert x_{Ca(i)})&=&\phi(\{\varepsilon: h_i(x_{Ca(i)}, \varepsilon_i)=x_i\}) \nonumber
\\
&=&\mu(x_i\vert do(x_{\pazocal{J}}), x_{Ca(i)\setminus\pazocal{J}}). \label{eq:th2eq2}
\end{eqnarray}
By Equations \ref{eq:th2eq1} and \ref{eq:th2eq2}, $\mu_{x_\pazocal{J}}(x_i\vert x_{Ca(i)\setminus\pazocal{J}})=\mu(x_i\vert do(x_{\pazocal{J}}), x_{Ca(i)\setminus\pazocal{J}})$.
Because $G$ represents $\overline{\succ}$, for each $x\in X$,
\begin{eqnarray*}
\mu_{x_{\pazocal{J}}}(x_{-\pazocal{J}}) &=&\Pi_{i\notin\pazocal{J}}\mu_{x_\pazocal{J}}(x_i\vert x_{Ca(i)\setminus \pazocal{J}}) \\ 
&=&\Pi_{i\notin\pazocal{J}}\mu(x_i\vert do(x_{\pazocal{J}}), x_{Ca(i)\setminus\pazocal{J}})\\
&=& \mu(x_{-\pazocal{J}}\vert do(x_{\pazocal{J}})).
\end{eqnarray*}
Thus, $\mu_{x_{\pazocal{J}}}(\cdot)=\mu(\cdot\vert do(x_\pazocal{J}))\in \Delta(X_{-\pazocal{J}})$. The converse implication --that $\mu_{x_{\pazocal{J}}}(\cdot)=\mu(\cdot\vert do(x_\pazocal{J}))\in \Delta(X_{-\pazocal{J}})$ for all $\pazocal{J}\subset\pazocal{N}$-- is a direct consequence of \autoref{th:newtheorem2}: The do-probability representation implies the rules of causal calculus, which implies \autoref{ax:axiom6}.

\indent Alternatively, \cite{pearl1995causal} shows that do-probabilities are characterized by the following condition for all $\pazocal{J}\subset\pazocal{N}$:
\begin{eqnarray}
P(x_{-\pazocal{J}}\vert do(X_{\pazocal{J}} = x_{\pazocal{J})}) &=& \frac{\Pi P(x_i \vert x_{Pa(i)})}{ \Pi_{j\in\pazocal{J}} P(x_j\vert x_{Pa(j)}) }, \label{dop2} \\
P(x_{-\pazocal{J}}\vert do(X_{\pazocal{J}} = x_{\pazocal{J})}) &=&  {\Pi_{i\notin \pazocal{J} } P(x_i \vert x_{Pa(i)}}), \label{dop3}
\end{eqnarray}

\indent The intuition behind \autoref{dop2} is that the do-probability $P(\cdot\vert x_{\pazocal{J}})$ is meant to eliminate $x_{\pazocal{J}}$ from the statistical model, while still retaining the ``direct'' effects of $\pazocal{J}$ on the remaining variables. Given the factorization formula for $P$, this means that the quotient in \autoref{dop2} reduces to \autoref{dop3}, where all instances of $P(x_\pazocal{J}\vert x_{Pa(j)}) $ are eliminated for each $j\in\pazocal{J}$, while still allowing variables in $\pazocal{J}$ to affect other variables through the various conditional terms $x_{Pa(\cdot)}$. \autoref{lemma:do-probability} shows that \autoref{ax:axiom6} implies \autoref{dop2}. Because \autoref{dop2} implies the rules of causal calculus \citep{pearl1995causal}, and the rules of causal calculus imply \autoref{ax:axiom4} (\autoref{th:newtheorem2}), then \autoref{ax:axiom6}, \autoref{dop2}, and Rules \ref{rule:rule1} and \ref{rule:rule2} are all equivalent. 

\subsection{The ``causal effect'' of $i$ on $j$, and more applications of \autoref{th:newtheorem2}}\label{app:causal-effect}

\indent Consider the example in the introduction: Alex wants to understand what ``the causal effect'' of education on earnings is. Our discussions throughout the paper suggest at least two ways in which the phrase ``the causal effect of education on earnings'' can be understood; for clarity, we call these the ``direct causal effect'' and the ``indirect causal effect''. In this section we comment of different ways we can define ``causal effects'', and how they differ.
\paragraph{The direct causal effect of $i$ on $j$} Bayesian causality (\autoref{def:causality}) suggests a natural way to define the causal effect of $i$ on $j$. For a fixed value of the variables other than $i$ and $j$, $x_{-\{i,j\}}\in X_{-\{i,j\}}$, let  $\Gamma_{x_{-\{i,j\}}}: X_{i}\rightarrow \Delta(X_j)$, where $\Gamma_{x_{-\{i,j\}}}(x_i)=\mu_{x_{i}, x_{-\{i,j\}}}(x_j)$. \autoref{def:causality} states that $i$ causes $j$ if \emph{ceteris paribus} interventions on $i$ affect the distribution of $j$, so it is natural to define causal effects as how sensitive the distribution of $j$ is to interventions on $i$, for each fixed intervention of the variables in $\{i,j\}^\complement$. We call this the \emph{direct} causal effect of $i$ on $j$.
\paragraph{The indirect causal effect of $i$ on $j$} Alternatively, we might want to consider the effect an intervention on $i$ has on $j$, \emph{without intervening any other variables}. For example, a government that believes education is a common cause of lifetime earnings and ability, and that ability is a further cause of lifetime earnings (see \autoref{fig:appendix-example}). The government wants to know how lifetime earnings change if the government imposed a policy that forced everyone to get a high-school degree. In contrast to the previous definition, the government is not interested in understanding how $\mu_{x_E, x_A}\in\Delta(X_L)$ changes with $x_E$ for a fixed $x_A$, because the government is not keeping ability levels constant. On the contrary, the government understands that the education policy will have two effects on lifetime earnings: the direct effect $E\rightarrow L$ and the indirect effect $E\rightarrow A \rightarrow L$. The government is therefore interested in understanding how $\mu_{x_E}\in \Delta(X_L)$ changes with $x_E$ \emph{without intervening any other variable}. In general, a policy maker might be interested in how $\mu_{x_i}\in\Delta(X_j)$ changes with $x_i$, which is generally different than the direct causal effect described in the previous paragraph. 

\begin{figure}[H]
\centering
\begin{tikzpicture}[thick]
\node[label=above:{E}] (v0) at (0,2) {};
\node[label=left:{A}] (v1) at (-2,0) {};
\node[label=right:{L}] (v2) at (2,0) {};
\draw [->] (v0)--(v1);
\draw [->] (v0)--(v2);
\draw [->] (v1)--(v2);
\end{tikzpicture}
\caption{Education as a common cause of Ability and Lifetime earnings}\label{fig:appendix-example}
\end{figure}

\indent More generally, a researcher may wonder what is the effect that a policy on variable $j$ will haver on variable $i$, in the presence of other policies (say, on variables $\pazocal{K}\subset\pazocal{N}$). In this case, the researcher calculates how $\mu_{x_{\pazocal{K}}, x_j }(\cdot)\in\Delta(X_i)$ varies with $x_j$ for each fixed level $x_{\pazocal{K}}$. This is yet another way in which $i$ has an indirect causal effect on $j$.

\indent Below, we show two simple examples to illustrate the use of DAGs to identify indirect causal effects.

\paragraph{A simple example} Consider a DM's causal DAG, depicted in \autoref{fig:confounder} below. To identify the indirect causal effect, $\mu_{x_E}(x_{D})$ the DM needs to understand which variables mediate information between $E$ and $D$. In other words, which variables would the DM be willing to pay in exchange for observing their realization? Because $Z$ is a common cause of $E$ and $D$, this is a variable the DM will need to condition on. Furthermore, both $A$ and $B$ are indirect common causes of $E$ and $D$. Therefore, at least one, perhaps both, needs to be conditioned on.

\begin{figure}[H]
\centering
\begin{tikzpicture}[thick]
\node[label=above:{A}] (v0) at (-2.20,1.52){};
\node[label=above:{B}] (v1) at (1.40,1.46) {};
\node[label=below right:{\textcolor{blue}{D}}] (v2) at (1.40,-1.62) {};
\node[label=below left:{\textcolor{red}{E}}]  (v3) at (-2.20,-1.60) {};
\node[minimum size=5mm] (v4) at (-0.300,0.0820) {Z};
\draw [->] (v0) edge (v3);
\draw [->] (v0) edge (v4);
\draw [->] (v1) edge (v2);
\draw [->] (v1) edge (v4);
\draw [->] (v3) edge (v2);
\draw [->] (v4) edge (v2);
\draw [->] (v4) edge (v3);
\end{tikzpicture}
\caption{And extended confounding triangle}\label{fig:confounder}
\end{figure}

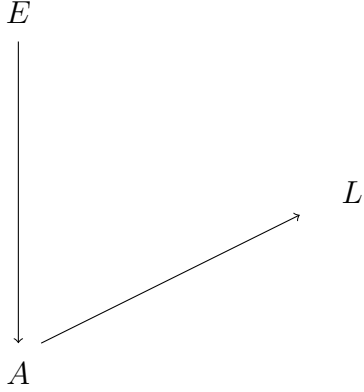
 We begin by noting that $\mu_{x_E}(x_D)=\sum \mu_{x_E}(x_D\vert x_A, x_Z)\mu_{x_E}(x_A,x_Z)=\sum\mu_{x_E}(x_D\vert x_A, x_Z)\mu_{x_E}(x_Z\vert x_A)\mu_E(x_A)$. Rule \ref{rule:rule2} implies $\mu_E(x_A)=\mu(x_A)$: once we remove the incoming arrows into $E$, the only paths connecting $E$ and $A$ are $E\rightarrow D\leftarrow Z\leftarrow A$ and $E\rightarrow D\leftarrow B\rightarrow Z\leftarrow A$. However, both these paths imply that $E$ and $A$ are independent.\footnote{The truncated DAG, $G_{E^{in}}$, represents any distribution $\mu$ that can be factorized as $\mu(A)\mu(B)\mu(E)\mu(Z\vert A,B)\mu(D\vert Z,B)$. More generally, see the D-separation criteria for a simple graphical tests of independence. In this case, because any undirected path joining $A$ to $E$ has head-to-head arrows, then $E$ and $A$ are independent.} Analogously, Rule \ref{rule:rule2} implies $\mu_E(x_Z\vert x_A)=\mu(x_Z\vert x_A)$. Finally, rule \ref{rule:rule1} implies $\mu_{x_E}(x_D\vert x_A, x_Z)=\mu(x_D\vert x_a, x_Z, x_E)$: if we eliminate the outgoing arrows from $E$, $E$ and $D$ are conditionally independent because any path joining $E$ to $D$ passes through either $Z$ or $A$. Therefore, the indirect causal effect of $E$ on $D$ is identified via the following condition:
\begin{eqnarray*}
\mu_{x_E}(x_D)=\sum \mu(x_D\vert x_A, x_Z, x_E)\mu(x_A,x_Z).
\end{eqnarray*}

\indent Finally, notice we could have carried out the same logic by using $B$ and not $A$ as our conditioning variable. This is useful in cases where some variables might be unobservable. In this case, the DM might understand that $B$ is a relevant variable, and is thus included in the model, but perhaps the DM cannot directly observe the values $B$ takes. As shown by the previous calculation this is not a problem because $\mu_{x_E}(x_D)$ can be identified from variables other than $B$.

\paragraph{Another example} The following DAG is taken from \cite{sebastiani2005genetic}. It shows a causal DAG that explores the connections between various genes and various disorders that these genes are associated with. Suppose that we were interested in the indirect causal effect of gene EDN1.3 (in red) on EDN1.7 (in blue). A priori, there are many paths through which these variables convey information on others. Therefore, many confounding paths could exist between one variable and another. Conveying all of these possible confounding paths in words would be, to put it colloquially, a nightmare. Furthermore, at face value it seems hard to understand which variables must be controlled for in order to elicit the desired causal effect. In other words, given the amount of causal links, it seems hard to apply the identification formulas in \autoref{th:newtheorem2}.

\begin{figure}[H]
\centering
\scalebox{0.5}{
\begin{tikzpicture}
\node (v0) at (-2.64,-16.6) {\textcolor{red}{EDN1.3}};
\node (v1) at (-1.23,-18.7) {\textcolor{blue}{EDNI1.7}};
\node (v2) at (-0.481,-2.64) {SELP.22};
\node (v3) at (-1.37,-3.94) {SELP.17};
\node (v4) at (-0.280,-7.26) {ECE1.13};
\node (v5) at (1.20,-8.77) {ECE1.12};
\node (v6) at (3.56,-11.2) {MET.6};
\node (v7) at (4.50,-12.8) {CAT};
\node (v8) at (2.51,8.20) {CSF2.3};
\node (v9) at (-1.28,8.20) {CSF2.4};
\node (v10) at (2.90,3.70) {TGFBR3.10};
\node (v11) at (1.35,2.06) {TGFBR3.2};
\node (v12) at (-0.0920,5.52) {TGFBR3.8};
\node (v13) at (3.14,-5.23) {SELP.12};
\node (v14) at (-3.21,4.79) {BMP6.13};
\node (v15) at (-2.38,7.07) {BMP6.12};
\node (v16) at (3.99,-20.1) {\textcolor{purple}{ANXA2.8}};
\node (v17) at (-1.18,0.789) {BMP6.11};
\node (v18) at (2.20,-9.86) {ADCY9.8};
\node (v19) at (-3.09,-2.25) {BMP6.10};
\node (v20) at (-1.99,-4.56) {BMP6};
\node (v21) at (0.105,-12.6) {ANXA2.11};
\node (v22) at (1.44,-15.4) {ANXA2.13};
\node (v23) at (4.32,-3.49) {MET.5};
\node (v24) at (-2.42,2.35) {BMP6.9};
\node (v25) at (-1.42,-9.21) {ANXA2.7};
\node (v26) at (-4.25,6.61) {BMP6.14};
\node (v27) at (4.45,7.34) {TGFBR3.7};
\node (v28) at (-1.33,-13.5) {ANXA2.12};
\node (v29) at (3.65,1.13) {SELP.14};
\node (v30) at (1.42,-3.29) {Stroke};
\node (v31) at (-3.41,-11.2) {ANXA2.5};
\node (v32) at (1.49,6.46) {TGFBR3.9};
\node (v33) at (1.70,-21.8) {EDNI1.6};
\node (v34) at (-4.06,-21.1) {\textcolor{purple}{EDN1.9}};
\node (v35) at (-0.282,-22.3) {EDN1.10};
\draw [->] (v3) edge (v2);
\draw [->] (v4) edge (v3);
\draw [->] (v5) edge (v4);
\draw [->] (v6) edge (v13);
\draw [->] (v7) edge (v23);
\draw [->] (v8) edge (v9);
\draw [->] (v10) edge (v11);
\draw [->] (v10) edge (v32);
\draw [->] (v11) edge (v32);
\draw [->] (v12) edge (v11);
\draw [->] (v12) edge (v32);
\draw [->] (v0) edge (v28);
\draw [very thick, ->] (v0) edge (v1);
\draw [->] (v14) edge (v15);
\draw [->] (v14) edge (v24);
\draw [->] (v35) edge (v1);
\draw [->] (v33) edge (v35);
\draw [->] (v33) edge (v1);
\draw [->] (v16) edge (v0);
\draw [->] (v16) edge (v7);
\draw [->] (v16) edge (v5);
\draw [->] (v16) edge (v33);
\draw [->] (v16) edge (v18);
\draw [->] (v16) edge (v20);
\draw [->] (v16) edge (v22);
\draw [->] (v17) edge (v24);
\draw [->] (v19) edge (v17);
\draw [->] (v19) edge (v14);
\draw [->] (v20) edge (v19);
\draw [->] (v22) edge (v21);
\draw [->] (v23) edge (v6);
\draw [->] (v23) edge (v27);
\draw [->] (v23) edge (v29);
\draw [->] (v34) edge (v35);
\draw [->] (v34) edge (v26);
\draw [->] (v34) edge (v0);
\draw [->] (v25) edge (v21);
\draw [->] (v25) edge (v31);
\draw [->] (v26) edge (v15);
\draw [->] (v27) edge (v8);
\draw [->] (v27) edge (v10);
\draw [->] (v28) edge (v22);
\draw [->] (v28) edge (v21);
\draw [->] (v28) edge (v25);
\draw [->] (v28) edge (v31);
\draw [->] (v30) edge (v9);
\draw [->] (v30) edge (v10);
\draw [->] (v30) edge (v6);
\draw [->] (v30) edge (v13);
\draw [->] (v30) edge (v2);
\draw [->] (v30) edge (v4);
\draw [->] (v30) edge (v5);
\draw [->] (v30) edge (v15);
\draw [->] (v30) edge (v14);
\draw [->] (v30) edge (v17);
\draw [->] (v30) edge (v18);
\draw [->] (v30) edge (v19);
\draw [->] (v30) edge (v12);
\draw [->] (v30) edge (v23);
\draw [->] (v30) edge (v29);
\end{tikzpicture}}
\caption{Causal DAG describing the associations of SNPs in candidate genes with the likelihood of developing nonhemorrhagic stroke in sickle cell anemia \citep{sebastiani2005genetic}.}\label{fig:sebastiani}
\end{figure}

\indent \indent This example highlights two advantages of the DAG language, and the rules of causal calculus. First, DAGs provide a succinct and transparent way to represent all causal interactions in the DM's model of interest. \autoref{th:theorem1} provides the exact conditions on behavior so that DAGs can indeed be a representation of a DMs causal model. Second, the rules of causal calculus are easily programable into web-based apps. For instance, analyzing this graph by applying the rules of causal calculus reveals that only 2 variables (marked in maroon in \autoref{fig:sebastiani}) need to be controlled in order to identify both the causal effect of $EDN1.3$ on $EDN1.7$ and the analogous indirect causal effect. Having identified this, it is straightforward to use the rules in \autoref{th:newtheorem2} to identify the intervention belief $\mu_{x_{EDN1.3}}(x_{EDN1.7})$ like we did in \autoref{fig:confounder}. In other words,  the rules in \autoref{th:newtheorem2} can be applied with computational ease.
  
\end{document}